\documentclass[article]{article}
\usepackage{color}
\usepackage{graphicx}
\usepackage{amsmath}
\usepackage{amssymb}
\usepackage{mathrsfs}
\usepackage[numbers,sort&compress]{natbib}
\usepackage{amsthm}
\numberwithin{equation}{section}
\usepackage{pgf}
\usepackage{CJK}
\usepackage{graphicx}
\usepackage{tikz}
\usepackage{stmaryrd}
\usepackage{graphicx}
\usepackage{hyperref}
\usepackage[latin1]{inputenc}
\usepackage[T1]{fontenc}
\usepackage[english]{babel}
\usepackage{listings}
\usepackage{xcolor,mathrsfs,url}
\usepackage{amssymb}
\usepackage{amsmath}
\usepackage{ifthen}
\newtheorem{theorem}{Theorem}
\newtheorem{lemma}{Lemma}
\newtheorem{corollary}{Corollary}
\newtheorem{proposition}{Proposition}
\newtheorem{Assumption}{Assumption}
\usepackage {caption}
\usepackage{float}
\usepackage{subfigure}

\DeclareMathOperator*{\res}{Res}

\begin{document}

\title{ Long-time asymptotic behavior of the modified   Schr\"{o}dinger  equation via  $\bar{\partial}$-steepest descent method  }
\author{Yiling YANG$^1$ and Engui FAN$^{1}$\thanks{\ Corresponding author and email address: faneg@fudan.edu.cn } }
\footnotetext[1]{ \  School of Mathematical Sciences, Fudan University, Shanghai 200433, P.R. China.}

\date{ }

\maketitle
\begin{abstract}
	\baselineskip=17pt

In this paper, we consider  the Cauchy problem for      the modified NLS equation
\begin{align}
&iu_t+u_{xx}+2\rho|u|^2u+i(|u|^2 u)_x=0,  \nonumber\\
&u(x,0)=u_0(x)\in H^{2,2}(R),\nonumber
\end{align}	
 where   $H^{2,2}(R)$ is  a weighted Sobolev space.
Using    nonlinear steepest descent method  and   combining   the $\bar{\partial}$-analysis,
  we show  that inside any fixed cone
\begin{equation}
C(x_1,x_2,v_1,v_2)=\left\lbrace (x,t)\in R^2|x=x_0+vt, x_0\in[x_1,x_2]\text{, }v\in[v_1+4\rho,v_2+4\rho]\right\rbrace.\nonumber
\end{equation}
the long time asymptotic behavior of the solution  $u(x,t)$   for the modified NLS equation  can
be characterized with  an $N(I)$-soliton   on discrete spectrum  and    leading order term $ O(t^{-1/2}) $  on continuous   spectrum  up to an residual error order $O(t^{-3/4})$.
\end{abstract}

\baselineskip=17pt

\newpage

\section {Introduction}
In this paper we study the long time asymptotic behavior for   the initial value problem  of the  modified NLS
  equation
\begin{equation}
iu_t+u_{xx}+2\rho|u|^2u+i(|u|^2 u)_x=0, \hspace{0.5cm} u(x,0)=u_0(x), \label{MNLS}
\end{equation}	
where  $\rho\in R$   and the   initial data  $u_0(x)$  belongs to   the weighted Sobolev space
\begin{equation*}
H^{2,2}(R)=\left\lbrace f\in L^2(R);x^2f,f''\in L^2(R)\right\rbrace.
\end{equation*}
The modified NLS equation  (\ref{MNLS})   was proposed to describe the nonlinear propagation of the Alfv$\grave{e}$n waves, the femtosecond optical pulses in a nonlinear single-mode optical fiber and the deep-water gravity waves \cite{Mio1,Stiassnie1}.  The term $i(|u|^2u)_x$ in the  equation (\ref{MNLS})  is called the self-steepening term, which causes an optical pulse to become asymmetric and steepen upward at the trailing edge \cite{Al, Yang}.  The equation (\ref{MNLS}) also describes the short pulses propagate in a long optical fiber characterized by a nonlinear refractive index \cite{Nakatsuka1,Tzoar1}.
Brizhik et al  showed  that the modified NLS equation (\ref{MNLS}), unlike the classical  NLS equation
\begin{equation}
iu_t+u_{xx}+2|u|^2u=0, \label{nls}
\end{equation}
possesses static localized solutions when the effective nonlinearity parameter is larger than a certain critical value \cite{RN13}.
  1970s, Wadati et al  showed that the equation  (\ref{MNLS})  is  completely integrable by  inverse scattering transformation  \cite{Wadati1979}. In recent years,  various exact  solutions for  the equation  (\ref{MNLS}) also  has  been extensively discussed  by   analytical method,  Hirota bilinear method and   Darboux transformation respectively \cite{DM1993, yangxiao, RN59, wenXY}.
The Hamiltonian structure from  mathematical structures  for the equation (\ref{MNLS}) was given \cite{TK2011}.   The inverse transformation  and dressing method  were used to
 construct N-soliton solutions of the modified  NLS  equation  (\ref{MNLS}) with zero boundary conditions were considered  \cite{chen1990,chen1991,Doktorov}.   Recently,  we  presented   inverse transformation
  for the  modified  NLS  equation  (\ref{MNLS}) with nonzero boundary conditions   by using Riemann-Hilbert  method \cite{yangfan2019}.
 From the determinant expressions of N-soliton solutions of the   modified NLS equation (\ref{MNLS}),  the   asymptotic behaviors of the N-soliton solutions  in the case of  $t \rightarrow\infty$  was  directly derived \cite{RN11}.
Kitaev and   Vartanian  applied  Deift-Zhou method  to  obtain  long-time  asymptotic solution  of  this equation (\ref{MNLS}) with decaying initial value.
They  derived an explicit functional form for the next-to-leading-order $O(t^{-1/2})$ term, that is \cite{AVKitaev},
\begin{align}
&u(x,t)=\frac{c}{\sqrt{t}} + O(t^{-1} \text{log} t),\label{pop}
\end{align}
 where  $c$ is related to initial value and phase point.

The  study  on the long-time behavior of nonlinear wave equations solvable by the inverse scattering method was first carried out by Manakov \cite{Manakov1974}.  Zakharov and Manakov   give the first result   for large-time asymptotic  of solutions for the  NLS equation (\ref{nls}) with  decaying initial value  in 1976 \cite{ZM1976}.    The inverse scattering method    also    worked  for long-time behavior of integrable systems    such as  KdV,  Landau-Lifshitz  and the reduced Maxwell-Bloch   system \cite{SPC,BRF,Foka}.   In 1993,    Deift and Zhou developed a  nonlinear steepest descent method to rigorously obtain the long-time asymptotics behavior of the solution for the MKdV equation
  by deforming contours to reduce the original  Riemann-Hilbert problem to a  model one  whose solution is calculated in terms of parabolic cylinder functions \cite{RN6}.    Later this method
  was applied  to  the focusing NLS equation, KdV equation, Fokas-Lenells equation and derivative NLS equation etc. \cite{RN9,RN10,Grunert2009,xu2015,xufan2013}.

In recent years,   McLaughlin and   Miller further  presented a $\bar\partial$ steepest descent method to analyze asymptotic of orthogonal polynomials with non-analytical weights.  This
   method   combine    steepest descent  with  $\bar{\partial}$-problem  rather than the asymptotic analysis of singular integrals on contours \cite{MandM2006,MandM2008}.
  When  it  is applied  to integrable systems,   the $\bar\partial$ steepest descent method  also has advantages,  such as   avoiding     delicate estimates involving $L^p$ estimates  of Cauchy projection operators, and leading  the non-analyticity in the RHP reductions to a $\bar{\partial}$-problem in some sectors of the complex plane  which can be solved by being recast into an integral equation and by using Neumann series.   Dieng and  McLaughin use it to study the defocusing NLS equation  under essentially minimal regularity assumptions on finite mass initial data \cite{DandMNLS};  Cussagna and  Jenkins study the defocusing NLS equation with finite density initial data \cite{SandRNLS}; They were also successfully applied to prove asymptotic stability of N-soliton solutions to focusing NLS \cite{fNLS} which  has been conjectured for a long time \cite{VEZ1972};  Jenkins and Liu study the derivative nonlinear Schr$\ddot{o}$dinger equation for generic initial data in a weighted Sobolev space \cite{Liu3}. Their work decomposes the solution into the sum of a finite number of separated solitons and a radiative partas when $t\to\infty$. And  the dispersive part contains two components, one coming from the continuous spectrum and another from the interaction of the discrete and continuous spectrum.

In our paper, we obtain the long-time asymptotic behavior of the solution of modified NLS equation (\ref{MNLS}) with initial data $u_0\in H^{2,2}$   by using the steepest descent method   and  $\bar\partial$ steepest descent method.   In  the recent work on the focusing NLS equation,  Borghese, Jenkins and McLaughlin  showed  how to treat a problem with discrete and  continuous spectral data \cite{fNLS};  The work on derivative NLS  equation due to  Jenkins and Liu  is the special case of  modified NLS with $\rho=0$  \cite{Liu3}.

This paper is arranged as follows.   Following the idea in  \cite{Liu3},    we reduce  the \textbf{RHP 1} in following context  into two parts,  one describes the asymptotic behavior of solitons,  and another model  computes the contributions due to the interactions of solitons and radiation.    In section 2,  we first make gauge transformation
\begin{equation*}
u=qe^{-i\int_{-\infty}^x |q|^2dy}, \hspace{0.5cm}u_0=q_0e^{-i\int_{-\infty}^x |q_0|^2dy},
\end{equation*}
and change the modified NLS equation (\ref{MNLS}) into  a new  equation which is more convenient to be dealt with.  We  describe the forward scattering transform step and  necessary results,
and establish  the  inverse scattering transform  with a  vector \textbf{RHP 1}.  In section 3, we   define a new row-vector RHP $M^{(1)}$ by (\ref{transm1})  which  deforms  the contour $R$  such that    the jump matrix (\ref{jump3}) approaches the identity exponentially fast away from the critical point $z_0$ (see Figure \ref{fig1}).    In section 4, we  deform \textbf{RHP 2} on a new contour $\Sigma^{(2)}$ whose jump matrices approach the identity exponentially fast away from the critical point $z_0$ by  defining a new unknown $M^{(2)}$(\ref{transm2}) which solves a mixed $\bar{\partial}$-Riemann-Hilbert problem--\textbf{RHP 3}. In section 5, we decompose $M^{(2)}$   into a model Riemann-Hilbert problem--\textbf{RHP 4} with solution $M^{RHP}$ and a pure $\bar{\partial}$-Problem--\textbf{$\bar{\partial}$-problem 5} with solution $M^{(3)}$.   To  solve  $M^{RHP}$,  we divide   it into  an outer model $M^{(out)}$ for the soliton components in Section 6,   and an inner model $M^{(z_0)}$ for the stationary phase point $z_0$ which is constructed by $M^{pc}$ by parabolic cylinder functions  in Section 7.    The outer and inner models together with error $E(z)$ build $M^{RHP}$ in (\ref{transm4}), where $E(z)$ is solution of a  small-norm Riemann-Hilbert problem in Section 8. Then we solve the \textbf{$\bar{\partial}$-problem 5} for $M^{(3)}$ in Section 9.  Thus, combinng  previous result we obtain
\begin{align}
M(z)=&M^{(3)}(z)M^{RHP}(z)R^{(2)}(z)^{-1}T(z)^{\sigma_3},\nonumber
\end{align}
for  brevity we omitting the dependence of (x,t). Then from the  asymptotic behavior as $|t|\to\infty$ of every function and the reconstruction formula we get (\ref{resultq}). But our ultimate purpose is to get the long-time  asymptotic behavior of $u(x,t)$, so we need to  establish the asymptotic formula of $e^{-i\int_{-\infty}^x |q|^2dy}$. Unlike in the reference
\cite{Liu3}, which is  our  special case  with $\rho=0$, we need to calculate $M_+(\rho)$ and  study its long-time  asymptotic behavior. At the point $\rho$, $M(z)$ doesn't have such simple nice properties like it at $\rho=0$, for example, his $M(z)$ is continuous at 0 and its calculation is simple. We calculate the values of $M_+(\rho)$ the parts in corresponding section and combine all result to obtain a long-time estimation of the phase factor $e^{-i\int_{-\infty}^x |q|^2dy}$ in Section 10.  Compared with the result  (\ref{pop}),    we   get a more general result (\ref{longtime1})-(\ref{longtime2}), which not only  improve   error estimate  as  a sharp error  $O(t^{-3/4})$  for more general initial data  $u_0\in H^{2,2}$.

\section {The inverse scattering method}
The modified NLS equation (\ref{MNLS}) admits the Lax pair \cite{Foka}
\begin{equation}
\phi_x = L_0\phi,\hspace{0.5cm}\phi_t = M_0\phi, \label{lax0}
\end{equation}
where
\begin{equation}
L_0=-i (k^2-\rho)\sigma_3+kU,\nonumber
\end{equation}
\begin{equation}
M_0=-2i(k^2-\rho)^2\sigma_3+2 k(k^2-\rho)U-ik^2U^2\sigma_3-ikU_x\sigma_3+kU^3\sigma_3,\nonumber
\end{equation}
and
\begin{equation}
\hspace{0.5cm}\sigma_3=\left(\begin{array}{cc}
1 & 0   \\
0 & -1
\end{array}\right),\hspace{0.5cm}U=\left(\begin{array}{cc}
0 & u  \\
-\bar{u} &0
\end{array}\right).\nonumber
\end{equation}
To avoid the imaginary axis becoming its boundary, we use a new Lax pair which is written in terms of $z=k^2$. And we discard the symmetry of the equation to avoid the factor $z^{1/2}$.
So  we  make   transformation
\begin{align}
&\phi =\left(\begin{array}{cc}
k & 0   \\
0 & 1
\end{array}\right)\psi,\nonumber 
\end{align}
and  get Lax pair \cite{chen1991} 
\begin{equation}
\psi_x = L\psi,\hspace{0.5cm}\psi_t = M\psi, \label{lax1}
\end{equation}
where
\begin{equation}
L=-i (z-\rho)\sigma_3+\Lambda U,\nonumber
\end{equation}
\begin{equation}
M=-2i(z-\rho)^2\sigma_3+2 (z-\rho)\Lambda U-izU^2\sigma_3-i\Lambda U_x\sigma_3+\Lambda U^3\sigma_3,\nonumber
\end{equation}
and
\begin{equation}
\hspace{0.5cm}\Lambda =\left(\begin{array}{cc}
1 & 0   \\
0 & z
\end{array}\right).\nonumber
\end{equation}

In order to have good asymptotics at $x \to \pm\infty$, because of the necessity of that the  diagonal line of $\phi\to I$ as  $x \to \pm\infty$, we first make following  transformation
\begin{equation}
u=qe^{-i\int_{-\infty}^x |q|^2dy}, \hspace{0.5cm}u_0=q_0e^{-i\int_{-\infty}^x |q_0|^2dy}.\label{trans1}
\end{equation}
This nonlinear, invertible mapping: $u\to q$ is an isometry of $L^2(R)$, which maps soliton solutions to soliton solutions, and maps dense open sets to dense open sets in weighted Sobolev spaces. It has  inverse transformation
\begin{equation}
q=ue^{i\int_{-\infty}^x |u|^2dy}.
\end{equation}
Then equation (\ref{MNLS}) is gauge-equivalent to
\begin{equation}
iq_t+q_{xx}+2\rho|q|^2q+i(|q|^2 q)_x-2i|q|^2q_x=0, \label{MNLS1}
\end{equation}
with Lax pair transformation
\begin{equation}
\psi=e^{-i/2\int_{-\infty}^x |q|^2dy\sigma_3}\Psi,
\end{equation}
from which we obtaina new Lax pair
\begin{align}
&\Psi_x =  -i (z-\rho)\sigma_3\Psi+Q \Psi,\label{lax2.1}\\
&\Psi_t =  -2i(z-\rho)^2\sigma_3\Psi+P \Psi, \label{lax2.2}
\end{align}
where
\begin{equation}
Q=\Lambda Q_0+\dfrac{i}{2}q^2\sigma_3,\nonumber \hspace{0.5cm}Q_0=\left(\begin{array}{cc}
0 & q  \\
-\bar{q} &0
\end{array}\right),
\end{equation}
\begin{equation}
P=2 (z-\rho)\Lambda Q_0-izQ_0^2\sigma_3-i\Lambda (Q_0)_x\sigma_3+\Lambda Q_0^3\sigma_3.\nonumber
\end{equation}

First, we consider the situation  that $q$ only dependent on $x \in R$.
Consider the Jost solutions of (\ref{lax2.1}) , which are restricted by the boundary conditions that
\begin{equation}
\Psi_\pm \sim  e^{-i(z-\rho)\sigma_3x}, \hspace{0.5cm}x\to \pm\infty.\label{asyx}
\end{equation}
They can be expressed as
\begin{align*}
\Psi_\pm =\left(\begin{array}{cc}
\Psi^1_\pm & \hat{\Psi}^1_\pm   \\
\Psi^2_\pm & \hat{\Psi}^2_\pm
\end{array}\right).
\end{align*}
The first and second columns of  the equation (\ref{lax2.1}) can be written respectively as
\begin{align}
\left(\begin{array}{c}
\Psi^1_\pm    \\
\Psi^2_\pm
\end{array}\right)_x=-i(z-\rho)\sigma_3\left(\begin{array}{c}
\Psi^1_\pm    \\
\Psi^2_\pm
\end{array}\right)+Q\left(\begin{array}{c}
\Psi^1_\pm    \\
\Psi^2_\pm
\end{array}\right),\\
\left(\begin{array}{c}
\hat{\Psi}^1_\pm    \\
\hat{\Psi}^2_\pm
\end{array}\right)_x=-i(z-\rho)\sigma_3\left(\begin{array}{c}
\hat{\Psi}^1_\pm    \\
\hat{\Psi}^2_\pm
\end{array}\right)+Q\left(\begin{array}{c}
\hat{\Psi}^1_\pm    \\
\hat{\Psi}^2_\pm
\end{array}\right),
\end{align}
from which we can obtain
\begin{align}
&(\Psi^2_\pm )_x=i(z-\rho)\Psi^2_\pm-z\bar{q}\Psi^1_\pm-\dfrac{i}{2}|q|^2\Psi^2_\pm,\\
&(\Psi^1_\pm )_x=-i(z-\rho)\Psi^1_\pm+q\Psi^2_\pm+\dfrac{i}{2}|q|^2\dfrac{i}{2}|q|^2\Psi^1_\pm,\\
&(\hat{\Psi}^1_\pm )_x=-i(z-\rho)\hat{\Psi}^1_\pm+q\hat{\Psi}^2_\pm+\dfrac{i}{2}|q|^2\hat{\Psi}^1_\pm ,\\
&(\hat{\Psi}^2_\pm )_x=i(z-\rho)\hat{\Psi}^2_\pm-z\bar{q}\hat{\Psi}^1_\pm-\dfrac{i}{2}|q|^2\hat{\Psi}^2_\pm .
\end{align}
By simple calculation, we obtain
\begin{align}
&\overline{(-\dfrac{1}{z}\Psi^2_\pm )_x}=-i(\bar{z}-\rho)(\overline{-\dfrac{1}{z}\Psi^2_\pm})+q\overline{\Psi^1_\pm}+\dfrac{i}{2}|q|^2(\overline{-\dfrac{1}{z}\Psi^2_\pm}),\\
&\overline{(\Psi^1_\pm )_x}=i(\bar{z}-\rho)\overline{\Psi^1_\pm}-\bar{z}\bar{q}(\overline{-\dfrac{1}{z}\Psi^2_\pm})-\dfrac{i}{2}|q|^2\overline{\Psi^1_\pm}.\\
\end{align}
Comparing their asymptotic condition (\ref{asyx}) and equation respectively, we get the symmetry of $\Psi$ as follow
\begin{align}
\overline{\Psi^1_\pm(\bar{z})}=\hat{\Psi}^2_\pm ,\hspace{0.5cm}\overline{-\dfrac{1}{z}\Psi^2_\pm(\bar{z})}=\hat{\Psi}^1_\pm.\label{symPsi}
\end{align}
These  solutions can be expressed as Volterra type integrals
\begin{equation}
\left(\begin{array}{c}
\Psi^1_\pm    \\
\Psi^2_\pm
\end{array}\right)=e^{-i(z-\rho)\sigma_3}+\int_{x}^{\pm\infty}e^{i(z-\rho)\hat{\sigma}_3}Q\left(\begin{array}{c}
\Psi^1_\pm    \\
\Psi^2_\pm
\end{array}\right)dy.
\end{equation}
Then we obtain that $\Psi^1_+$ and $\Psi^2_+$ are analysis in $C^-$. In the same way we obtain that  $\Psi^1_-$ and $\Psi^2_-$ are analysis in $C^+$.

By making transformation
\begin{equation}
\varphi_\pm=\Psi_\pm e^{i(z-\rho)\sigma_3},\label{trans2}
\end{equation}
we then  have
\begin{equation*}
\varphi_\pm \sim I, \hspace{0.5cm} x \rightarrow \pm\infty.
\end{equation*}
Moreover,  $\varphi_\pm$    satisfy an  equivalent Lax pair
\begin{align}
&(\varphi_\pm)_x = \left(-i (z-\rho)[\sigma_3,\varphi_\pm]+Q\right)\varphi_\pm,\label{lax3.1}\\
&(\varphi_\pm)_t = \left(-2i(z-\rho)^2[\sigma_3,\varphi_\pm]+P\right)\varphi_\pm, \label{lax3.2}
\end{align}
For  convenience, we denote two elements of the first column of $\varphi_\pm$ as $\varphi^2_\pm$  and $\varphi^1_\pm$  respectively.
Since ${\rm tr}\left(-i (z-\rho)\sigma_3+Q\right)=0$  in  (\ref{lax2.1}) and (\ref{lax2.2}),  by using Able formula, we  have
\begin{equation}
(\det\Psi_\pm)_x=0. \label{abl}
\end{equation}
Again by using the relation
\begin{equation*}
\det(\varphi_\pm)=\det(\Psi_\pm e^{i(z-\rho)\sigma_3})=\det(\Psi_\pm),
\end{equation*}
we get  $(\det\varphi_\pm)_x=0$, which means  that  $\det(\varphi_\pm)$ is independent of $x$.  So we obtain that
\begin{equation}
\det\varphi_\pm=\lim_{x \to \pm\infty}\det(\varphi_\pm)=\det I=1, \label{detphi}
\end{equation}
which implies that $\varphi_\pm$ are inverse matrices.

Since   $\Psi_\pm$ are two fundamental matrix solutions of the  Lax  pair (\ref{lax2.1}) and (\ref{lax2.2}),  there exists a linear  relation between $\Psi_+$ and $\Psi_-$, namely
\begin{equation}
\Psi_+(x,z)=\Psi_-(x,z)S(z),\hspace{0.5cm} z\in R,\label{scattering}
\end{equation}
where $S(z)$ is called scattering matrix and  $\det S(z)=1$  by using  (\ref{detphi}). Form the symmetry of  $\Psi_\pm$ (\ref{symPsi}),   $S(z)$ can be written as
\begin{align*}
S(z) =\left(\begin{array}{cc}
a(z) & b(z)   \\[4pt]
-z\overline{b(z)} & \overline{a(z)}
\end{array}\right),
\end{align*}
then we have
\begin{equation}
|a(z)|^2+z|b(z)|^2=1.
\end{equation}
We introduce the reflection coefficient
\begin{equation}
r(z)=\dfrac{\overline{b(z)}}{a(z)}.
\end{equation}
Then we immediately have $1+z|r(z)|^2=|a(z)|^{-2}$.
Calculating individual elements of  the matrix equation (\ref{scattering}), then we obtain that the function a(z) and b(z) may be computed via the  Wronskian formula:
\begin{equation}
a(z)=\Psi_+^1\hat{\Psi}_-^2-\Psi_+^2\hat{\Psi}_-^1,\hspace{0.5cm}b(z)=\hat{\Psi}_+^1\hat{\Psi}_-^2-\hat{\Psi}_+^2\hat{\Psi}_-^1.\label{ab}
\end{equation}
Thus a(z) has an analytic continuation to $C^+$, but b(z) has no analyticity.

To  get   the Riemann-Hilbert problem , it is necessary to discuss the asymptotic behaviors of the Jost solutions and scattering matrix as $z \rightarrow \infty$. We denote 	We consider the following  asymptotic  expansions
\begin{align}
&\varphi_\pm(x,z)=\varphi_\pm^{(0)}(x)+\frac{\varphi_\pm^{(1)}(x)}{z}+\frac{\varphi_\pm^{(2)}(x)}{z^2}+O(z^{-3}),\hspace{0.5cm}\text{as }z \rightarrow \infty.\label{expansion1}
\end{align}
Substituting   (\ref{expansion1})   into the Lax pair (\ref{lax3.1}) and comparing the coefficients,  we obtain
\begin{align*}
&2i\varphi^{2,(0)}_\pm-\bar{q}\varphi^{1,(0)}_\pm=0,\\
&(\varphi^{1,(0)}_\pm)_x=\frac{i}{2}|q|^2\varphi^{1,(0)}_\pm+q\varphi^{2,(0)}_\pm.
\end{align*}
Then we obtain $(\varphi^{1,(0)}_\pm)_x=0$, which means
\begin{align*}
&\varphi^{1}_\pm \rightarrow 1,\hspace{0.5cm}\text{as }z \rightarrow \infty,\\
&\varphi^{2}_\pm \rightarrow -\dfrac{i}{2}\bar{q},\hspace{0.5cm}\text{as }z \rightarrow \infty.
\end{align*}

Therefore from (\ref{ab}) we have
\begin{align*}
&a(z)\rightarrow 1,\hspace{0.5cm}\text{as }z \rightarrow \infty.
\end{align*}
When $z$ is real, we have
\begin{align*}
b(z)=O(z),\hspace{0.5cm}r(z)=O(z),\hspace{0.5cm}\text{as }z \rightarrow \infty.
\end{align*}

Zeros of $a$ on R are known to occur and they correspond to spectral singularities \cite{RN3}. They are excluded from our analysis in the this paper. To deal with our following job, we assume our initial data satisfy this assumption.
\begin{Assumption}\label{initialdata}
	The initial data $u_0 \in H^{2,2}(R)$     and it generates generic scattering data which satisfy that
	
	\textbf{1. }a(z) has no zeros on $R$.
	
	\textbf{2. }a(z) only has finite number of simple zeros.
	
	\textbf{3. }a(z) and r(z) belong  $H^{2,2}(R)$.
\end{Assumption}
In the absence of spectral singularities (real zeros of $a(z)$), there also exist $c\in (0,1)$ such that $c<|a(z)|<1/c$ for $z\in R$, which implies $1+z|r(z)|>c^2>0$ for $z\in R$.

We assume that $a(z)$ has N simple zeros $z_n \in C^+$, n = 1, 2,.., N, $a(z_n) = 0$. Denote $Z$=$\left\lbrace z_n\right\rbrace ^N_{n=1}$ which is the set of the zeros of $a(z)$.  A standard result of the scattering theory have following trace formula for the transmission coefficient:
\begin{equation}
a(z)=\prod_{k=1}^N\dfrac{z-z_k}{z-\bar{z}_k}\exp\left(i\int _{R}\dfrac{k(s)ds}{s-z}\right),\label{az}
\end{equation}
where
\begin{equation}
k(s)=-\frac{1}{2\pi}\log(1+s|r(s)|^2)\label{ks}.
\end{equation}
Combining (\ref{scattering}) and (\ref{trans2}) we obtain
\begin{equation}
\varphi_+=\varphi_-e^{-(z-\rho)x\hat{\sigma}_3}S(z),
\end{equation}
from which we have
\begin{equation}
a(z)=\lim_{x \to -\infty}\varphi_+^1.
\end{equation}
Then rewrite (\ref{lax3.1}), we have
\begin{equation}
(\varphi_+)_x=(z-\rho)\left( \varphi_+\left(\begin{array}{cc}
i & 0   \\
0 & -i
\end{array}\right)-\left(\begin{array}{cc}
i & 0   \\
\bar{q} & -i
\end{array}\right)\varphi_+ \right)+\left(\begin{array}{cc}
\frac{i}{2}|q|^2 & q   \\
-\rho\bar{q} & \frac{-i}{2}|q|^2
\end{array}\right)\varphi_+.
\end{equation}
Note that $\lim_{x \to +\infty}\varphi_+=I$, as $z=\rho$ is a regular point of this system of equations, we have
\begin{align}
\varphi_+(x,t,\rho)=&e^{\frac{-i\sigma_3}{2}\int_{x}^{+\infty}|q(y,t)|^2dy}\nonumber\\
&\left(\begin{array}{cc}
1 & -\int_{x}^{+\infty}q(y,t) e^{i\int_{y}^{+\infty}|q(s,t)|^2ds}dy \\
\int_{x}^{+\infty}\bar{q}(y,t) e^{-i\int_{y}^{+\infty}|q(s,t)|^2ds}dy & 1
\end{array}\right).\label{intq}
\end{align}
Then we take $x\to\-\infty$ and obtain that
\begin{equation}
a(\rho)=e^{-\frac{i}{2}\int_{-\infty}^{+\infty}|q(y,t)|^2dy}=\prod_{k=1}^N\dfrac{\rho-z_k}{\rho-\bar{z}_k}exp\left(i\int _{R}\dfrac{k(s)ds}{s-\rho}\right)\label{arho}.
\end{equation}
Similarly we have
\begin{equation}
(\varphi_-)_x=(z-\rho)\left( \varphi_-\left(\begin{array}{cc}
i & 0   \\
0 & -i
\end{array}\right)-\left(\begin{array}{cc}
i & 0   \\
\bar{q} & -i
\end{array}\right)\varphi_- \right)+\left(\begin{array}{cc}
\frac{i}{2}|q|^2 & q   \\
-\rho\bar{q} & \frac{-i}{2}|q|^2
\end{array}\right)\varphi_-,
\end{equation}
and
\begin{align}
\varphi_-(x,t,\rho)=&e^{\frac{i\sigma_3}{2}\int_{-\infty}^{x}|q(y,t)|^2dy}\nonumber\\
&\left(\begin{array}{cc}
1 & -\int_{-\infty}^{x}q(y,t) e^{-i\int_{-\infty}^{y}|q(s,t)|^2ds}dy \\
\int_{-\infty}^{x}\bar{q}(y,t) e^{i\int_{-\infty}^{y}|q(s,t)|^2ds}dy & 1
\end{array}\right).\label{varphi-rho}
\end{align}

Then we begin to calculate residue conditions.
At any $z_n$, $(\Psi_+^1,\Psi_+^2)$ and $(\Psi_-^1,\Psi_-^2)$ are linearly dependent. Specifically, a constant $b_k$ exists such that:
\begin{equation*}
(\Psi_+^1,\Psi_+^2)=b_k(\Psi_-^1,\Psi_-^2).
\end{equation*}
Denote\emph{ norming constant} $c_k=b_k/a'(z_k)$, for initial data $q_0$, the collection D=$\left\lbrace  r(z),\left\lbrace z_k,c_k\right\rbrace^N_{k=1}\right\rbrace $ is called the \emph{scattering data} for $q_0$.
It is an elementary calcution to show that the sectionally meromorphic matrices defined as follow
\begin{equation}
M(x,z)=\Bigg\{\begin{array}{ll}
M^+=\left(\begin{array}{cc}
a(z)^{-1}\varphi_+^1, & \hat{\varphi}^1_-\\[4pt]
a(z)^{-1}\varphi_+^2, & \hat{\varphi}^2_-
\end{array}\right), &\text{as } z\in D^+,\\[6pt]
M^-=\left(\begin{array}{cc}
\varphi_-^1, & \overline{a(\bar{z})}^{-1}\hat{\varphi}_+^1\\[4pt]
\varphi_-^2, & \overline{a(\bar{z})}^{-1}\hat{\varphi}_+^2
\end{array}\right), &\text{as }z\in D^-,\\
\end{array}
\end{equation}
solves the following Riemann-Hilbert problem:

\textbf{Riemann-Hilbert problem 0}. Find a matrix-valued function z$\in C \rightarrow m(z;x)$ which satisfies:

$\bullet$ Analyticity: $M(x,z)$ is meromorphic in $\mathbb{C}\setminus R$ and has single poles;

$\bullet$ Jump condition: M has continuous boundary values $M_\pm$ on $R$ and
\begin{equation}
M^+(x,z)=M^-(x,z)V(z),\hspace{0.5cm}z \in R,\label{jump}
\end{equation}
where
\begin{equation}
V(z)=\left(\begin{array}{cc}
1+z|r(z)|^2 & -e^{-2i(z-\rho)x}r(z)\\
-e^{2i(z-\rho)x}z\overline{r(z)} & 1
\end{array}\right);
\end{equation}

$\bullet$ Asymptotic behaviors:There exists p independent of z that
\begin{align}
&M(x,z) = \left(\begin{array}{cc}
1 & 0\\
p & 1
\end{array}\right)+O(z^{-1}),\hspace{0.5cm}z \rightarrow \infty;\label{asymbehv1}
\end{align}

$\bullet$ Residue conditions: M has simple poles at each point in $ Z\bigcup \bar{Z}$ with:
\begin{align}
&\res_{z=z_n}M(z)=\lim_{z\to z_n}M(z)\left(\begin{array}{cc}
0 & 0\\
c_ne^{2i(z_n-\rho)x} & 0
\end{array}\right),\label{RES1}\\
&\res_{z=\bar{z}_n}M(z)=\lim_{z\to \bar{z}_n}M(z)\left(\begin{array}{cc}
0 & -z_n^{-1}\bar{c}_ne^{-2i(\bar{z}_n-\rho)x}\\
0 & 0
\end{array}\right).\label{RES2}
\end{align}
From the asymptotic behavior of the functions $\varphi_\pm$, we have following reconstruction formula:
\begin{equation}
q(x)=2i\lim_{z\to \infty}[zM(z;x)]_{12}.
\end{equation}

Now we are going to take into account the time. If $q$ also depends on t (i.e.
$q$ = $q(x, t))$, we can obtain the functions a and b as above for all times t $\in R$. Taking account of (\ref{lax2.2}), we have
\begin{align}
\left( a(z;t)\right) _t=0,\hspace{0.5cm}\left( b(z;t)\right) _t=-4i(z-\rho)^2 b(z;t).
\end{align}
Then we can obtain time dependences of scattering data which can be expressed as the following replacement
\begin{align}
&c(z_n)\rightarrow c(t,z_n)=c(0,z_n)e^{-4i(z_n-\rho)^2t},\\
&r(z)\rightarrow r(t,z)=r(0,z)e^{-4i(z-\rho)^2t}
\end{align}
In particular, if at time $t = 0$ the initial function $q(x, 0)$ produces N simple zeros $z_1$,...,$z_N$ of $a(z; 0)$ and if $q$ evolves accordingly to the (\ref{MNLS1}), then $q(x, t)$ will produce exactly the same N simple zeros at any other time $t \in   R$. Altogether the scattering data of a function $q$(x, t), which is a solution of  (\ref{MNLS1}), is given at time $t$ by
\begin{equation*}
\left\lbrace  e^{-4i(z-\rho)^2t}r(z),\left\lbrace z_k,e^{-4i(z_n-\rho)^2t}c_k\right\rbrace^N_{k=1}\right\rbrace,
\end{equation*}
where $\left\lbrace  r(z),\left\lbrace z_k,c_k\right\rbrace^N_{k=1}\right\rbrace$ are obtained from the initial data $q(x, 0) =
q_0(x)$.
Let us introduce the phase function:
\begin{equation}
\theta(z)=(z-\rho)\frac{x}{t}+2(z-\rho)^2.
\end{equation}
For convenience we denote $\theta_n=\theta(z_n)$.   The (time-dependent) inverse spectral problem is defined by the following Riemann-Hilbert problem (RHP MNLS) and reconstruction formula:

\textbf{RHP modified NLS\label{RHPMNLS}}. Find a matrix-valued function z$\in C \rightarrow m(z;x,t)$ which satisfies:

$\bullet$ Analyticity: $M(z;x,t)$ is meromorphic in $\mathbb{C}\setminus R$ and has single poles;

$\bullet$ Jump condition: M has continuous boundary values $M_\pm$ on $R$ and
\begin{equation}
M^+(z;x,t)=M^-(z;x,t)V(z),\hspace{0.5cm}z \in R,\label{jump1}
\end{equation}
where
\begin{equation}
V(z)=\left(\begin{array}{cc}
1+z|r(z)|^2 & -e^{-2i\theta t}r(z)\\
-e^{2it\theta t}z\overline{r(z)} & 1
\end{array}\right);
\end{equation}

$\bullet$ Asymptotic behaviors:There exists p independent of z that
\begin{align}
&M(z;x,t) \sim \left(\begin{array}{cc}
1 & 0\\
p & 1
\end{array}\right)+O(z^{-1}),\hspace{0.5cm}z \rightarrow \infty;\label{asymbehv2}
\end{align}

$\bullet$ Residue conditions: M has simple poles at each point in $ Z\bigcup \bar{Z}$ with:
\begin{align}
&\res_{z=z_n}M(z)=\lim_{z\to z_n}M(z)\left(\begin{array}{cc}
0 & 0\\
c_ne^{2i\theta_n t} & 0
\end{array}\right),\label{RES3}\\
&\res_{z=\bar{z}_n}M(z)=\lim_{z\to \bar{z}_n}M(z)\left(\begin{array}{cc}
0 & -z_n^{-1}\bar{c}_ne^{-2i\bar{\theta}_n t}\\
0 & 0
\end{array}\right).\label{RES4}
\end{align}
Reconstruction formula is
\begin{equation}
q(x,t)=2i\lim_{z\to \infty}[zM(z;x,t)]_{12}\label{reconsq}.
\end{equation}
Inserting the time dependence into (\ref{RES1}) and (\ref{RES2}) we end up exactly with (\ref{RES3}) and (\ref{RES4}). Summarized the method of (inverse) scattering works as follows:
\begin{align*}
\begin{matrix}\begin{array}{ccc}
q_0(x)=q(x,0) & \longrightarrow\longrightarrow\longrightarrow& \left\lbrace  r(z),\left\lbrace z_k,c_k\right\rbrace^N_{k=1}\right\rbrace\\
\hspace{0.5cm}\hspace{0.5cm}\hspace{0.55cm}\downarrow & \text{scattering data} & \downarrow\\
\text{solve} (\ref{MNLS1})\downarrow  & \hspace{0.5cm}\hspace{0.5cm}\hspace{0.5cm}  & \hspace{2.83cm}\downarrow\text{time dependences}\\
q(x,t) &\longleftarrow\longleftarrow\longleftarrow & \left\lbrace  e^{-4i(z-\rho)^2t}r(z),\left\lbrace z_k,e^{-4i(z_n-\rho)^2t}c_k\right\rbrace^N_{k=1}\right\rbrace\\
\hspace{0.5cm} & \text{solve RHP MNLS} & \hspace{0.5cm}
\end{array}\end{matrix}
\end{align*}
Because \textbf{RHP MNLS} is not properly normalized, we seek a row vector-valued solution:

\textbf{Riemann-Hilbert Problem 1}. Find a row vector-valued function z$\in C \rightarrow M(z;x,t)$ which satisfies:

$\bullet$ Analyticity: $M(z;x,t)$ is meromorphic in $\mathbb{C}\setminus R$ and has single poles;

$\bullet$ Jump condition: M has continuous boundary values $M_\pm$ on $R$ and
\begin{equation}
M^+(z;x,t)=M^-(z;x,t)V(z),\hspace{0.5cm}z \in R,\label{jump0}
\end{equation}
where
\begin{equation}
V(z)=\left(\begin{array}{cc}
1+z|r(z)|^2 & -e^{-2i\theta t}r(z)\\
-e^{2it\theta t}z\overline{r(z)} & 1
\end{array}\right);
\end{equation}

$\bullet$ Asymptotic behaviors:There exists p independent of z that
\begin{align}
&M(x,z) = \left(\begin{array}{cc}
1 & 0\\
p & 1
\end{array}\right)+O(z^{-1}),\hspace{0.5cm}z \rightarrow \infty;\label{asymbehv1}
\end{align}

$\bullet$ Residue conditions: M has simple poles at each point in $ Z\bigcup \bar{Z}$ with:
\begin{align}
&\res_{z=z_n}M(z)=\lim_{z\to z_n}M(z)\left(\begin{array}{cc}
0 & 0\\
c_ne^{2i\theta_n t} & 0
\end{array}\right),\label{RES5}\\
&\res_{z=\bar{z}_n}M(z)=\lim_{z\to \bar{z}_n}M(z)\left(\begin{array}{cc}
0 & -z_n^{-1}\bar{c}_ne^{-2i\bar{\theta}_n t}\\
0 & 0
\end{array}\right).\label{RES6}
\end{align}

\section{Conjugation}
\quad The long-time asymptotic analysis of RHP 1 is determined by the growth and decay of the exponential function $e^{2it\theta}$ appearing in both the jump relation and the residue conditions. In this section, we describe a new transform: $M(z)\to M^{(1)}(z)$, from which we make that the \textbf{RHP} is well behaved as $|t|\to \infty$ along any characteristic line. Let $z_0=-\frac{4x}{t}+\rho$ be the (unique) critical point of the phase fuction $\theta(z)$. Then we have
\begin{equation}
\theta = (z-\rho)(2z-4z_0+2\rho),\hspace{0.5cm}Re (2it\theta)=-8t\text{Imz}(\text{Re}z-z_0).\label{theta}
\end{equation}
Recall the partition $\Delta^\pm_{z_0,\eta}$ of $\left\lbrace 1,...,N\right\rbrace $ for $z_0 \in R$, $\eta$ = $sgn(t)$ defined as follow:
\begin{align*}
\Delta^+_{z_0,1}=\Delta^-_{z_0,-1}=\left\lbrace k \in \left\lbrace 1,...,N\right\rbrace  |Re(z_k)>z_0\right\rbrace, \\
\Delta^-_{z_0,1}=\Delta^+_{z_0,-1}=\left\lbrace k \in \left\lbrace 1,...,N\right\rbrace  |Re(z_k)<z_0\right\rbrace. \\
\end{align*}
This partition splits the residue coefficients $c_n$ in two sets which is shown in Figure. \ref{fig1}.
\begin{figure}[H]
	\centering
		\centering
	\subfigure[]{
		\begin{tikzpicture}[node distance=2cm]
	\filldraw[yellow,line width=2] (2.4,0.01) rectangle (0.01,2.4);
	\filldraw[yellow,line width=2] (-2.4,-0.01) rectangle (-0.01,-2.4);
	\draw[->](-2.5,0)--(2.5,0)node[right]{Re$z$};
	\draw[->](0,2.5)--(0,2.5)node[above]{$t<0$};
	\draw[->](0,0)--(-0.8,0);
	\draw[->](-0.8,0)--(-1.8,0);
	\draw[->](0,0)--(0.8,0);
	\draw[->](0.8,0)--(1.8,0);
	\coordinate (A) at (0.5,1.2);
	\coordinate (B) at (0.6,-1.2);
	\coordinate (G) at (-0.6,1.2);
	\coordinate (H) at (-0.5,-1.2);
	\coordinate (I) at (0,0);
	\fill (A) circle (0pt) node[right] {$|e^{2it\theta}|\to \infty$};
	\fill (B) circle (0pt) node[right] {$|e^{2it\theta}|\to 0$};
	\fill (G) circle (0pt) node[left] {$|e^{2it\theta}|\to 0$};
	\fill (H) circle (0pt) node[left] {$|e^{2it\theta}|\to \infty$};
	\fill (I) circle (1pt) node[below] {$z_0$};
	\label{zplane1}
	\end{tikzpicture}
}
	\subfigure[]{
		\begin{tikzpicture}[node distance=2cm]
		\filldraw[yellow,line width=2] (2.4,0.01) rectangle (0.01,-2.4);
		\filldraw[yellow,line width=2] (-2.4,-0.01) rectangle (-0.01,2.4);
		\draw[->](-2.5,0)--(2.5,0)node[right]{Re$z$};
		\draw[->](0,2.5)--(0,2.5)node[above]{$t>0$};
		\draw[->](0,0)--(-0.8,0);
		\draw[->](-0.8,0)--(-1.8,0);
		\draw[->](0,0)--(0.8,0);
		\draw[->](0.8,0)--(1.8,0);
		\coordinate (A) at (0.6,1.2);
		\coordinate (B) at (0.5,-1.2);
		\coordinate (G) at (-0.5,1.2);
		\coordinate (H) at (-0.6,-1.2);
		\coordinate (I) at (0,0);
		\fill (A) circle (0pt) node[right] {$|e^{2it\theta}|\to 0$};
		\fill (B) circle (0pt) node[right] {$|e^{2it\theta}|\to \infty$};
		\fill (G) circle (0pt) node[left] {$|e^{2it\theta}|\to \infty$};
		\fill (H) circle (0pt) node[left] {$|e^{2it\theta}|\to 0$};
		\fill (I) circle (1pt) node[below] {$z_0$};
		\label{zplane2}
		\end{tikzpicture}
	}
	\caption{In the yellow region, $|e^{2it\theta}|\to \infty$ when $t\to\pm\infty$ respectively. And in white region, $|e^{2it\theta}|\to 0$ when $t\to\pm\infty$ respectively}
	\label{fig1}
\end{figure}
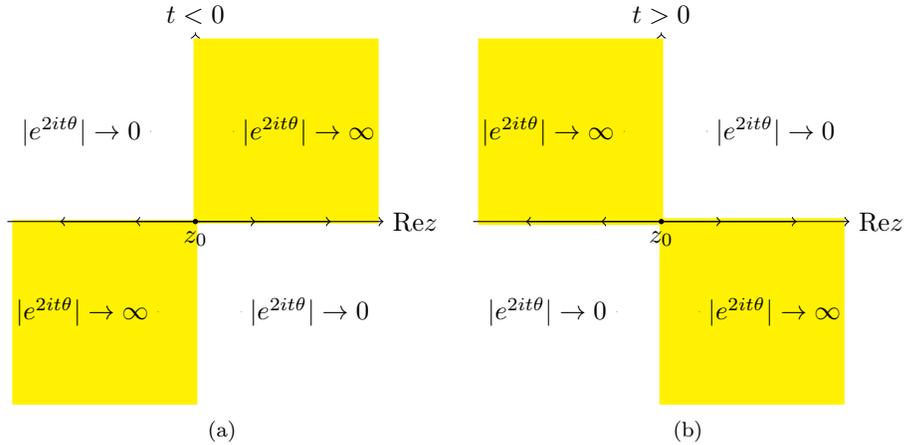
To introduce a transformation which renormalizes \textbf{RHP 1} with well conditioned fo $|t|\to \infty$ and fixed $z_0$, we denote following functions:
\begin{align}
&I_{z_0}^\eta=\left\lbrace  s\in R | -\infty <\eta s \leq \eta z_0\right\rbrace, \hspace{0.5cm}\delta (z)=\delta (z,z_0,\eta)=\exp\left(i\int _{I_{z_0}^\eta}\dfrac{k(s)ds}{s-z}\right)\\
&T(z)=T(z,z_0,\eta)=\prod_{k\in \Delta^-_{z_0,\eta}}\dfrac{z-\bar{z}_k}{z-z_k}\delta (z) \label{T}, \\
&\beta(z,z_0,\eta)=-\eta k(z_0)\log(\eta(z-z_0+1))+\int _{I_{z_0}^\eta}\dfrac{k(s)-X(s)k(z_0)}{s-z}ds,\\
&T_0=T(z_0,\eta)=\prod_{k\in \Delta^-_{z_0,\eta}}\dfrac{z_0-\bar{z}_k}{z_0-z_k}e^{i\beta (z_0,z_0,\eta)},
\end{align}
where $k(s)$ is defined in (\ref{ks}), and $X(s)$ is the characteristic function of the interval $\eta z_0-1<\eta s<\eta z_0$. In all of the above formulas, we choose the principal branch of power and logarithm functions.
From (\ref{az}) we find that the function $T(z,z_0)$ is a partial transmission coefficient which approaches $a(z)^{-1}$ as $z_0\to \eta$.
\begin{proposition}\label{proT}
	The function defined by (\ref{T}) has following properties:\\
	(a) $T$ is meromorphic in $C\setminus I_{z_0}^\eta$, for each $n\in\Delta^-_{z_0,\eta}$, $T(z)$ has a simple pole at $z_n$ and a simple zero at $\bar{z}_n$;\\
	(b) For $z\in C\setminus I_{z_0}^\eta$, $\overline{T(\bar{z})}T(z)=1$;\\
	(c) For $z\in I_{z_0}^\eta$, as z approaches the real axis from above and below, $T$ has boundary values $T_\pm$, which satisfy:
	\begin{equation}
	T_+(z)=(1+z|r(z)|^2)T_-(z),\hspace{0.5cm}z\in I_{z_0}^\eta;
	\end{equation}
	(d) As $|z|\to \infty$ with $|arg(z)|\leq c<\pi$,
	\begin{equation}
	T(z)=1+\frac{i}{z}\left[ 2\sum_{k\in \Delta^-_{z_0,\eta}} Im(z_k)-\int _{I_{z_0}^\eta}k(s)ds\right] \label{expT};
	\end{equation}
	(e) As $z\to z_0$, along $z=z_0+e^{i\psi}l$, $l>0$, $|\psi|\leq c<\pi$,
	\begin{equation}
	|T(z,z_0,\eta)-T_0(z_0,\eta)(\eta(z-z_0))^{i\eta k(z_0)}|\leq C|z-z_0|^{1/2}.
	\end{equation}
\end{proposition}
\begin{proof}
	Prats (a) and (b) are elementary consequences of the definition (\ref{T}). And for part (c), as it definition we only need to consider $\int _{I_{z_0}^\eta}\dfrac{k(s)ds}{s-z}$, and via the residue theorem we obtain the consequence. For part (d), we  expand  the product term and the factor $(s-z)^{-1}$ for large $z$
	\begin{align}
	&\prod_{k\in \Delta^-_{z_0,\eta}}\dfrac{z-\bar{z}_k}{z-z_k}=\prod_{k\in \Delta^-_{z_0,\eta}}\left[ 1+\frac{2i}{z}Im(z_k)+O(z^{-2})\right] ,\label{expprod}\\
	&\delta (z)=exp\left(-\frac{i}{z}\int _{I_{z_0}^\eta}k(s)ds+O(z^{-2})\right)=1-\frac{i}{z}\int _{I_{z_0}^\eta}k(s)ds+O(z^{-2}).\label{expdelta}
	\end{align}
	Noting the fact that $\parallel k\parallel_{L^1(R)}\leq (2\pi)^{-1}\parallel r\parallel_{H^{2,2}(R)}$, the integration in (\ref{expdelta}) also make sense. Combine (\ref{expprod}) and (\ref{expdelta}) we obtain (\ref{expT}) immediately. For part (e), we rewrite T:
	\begin{align*}
	T(z,z_0)=\prod_{k\in \Delta^-_{z_0,\eta}}\dfrac{z-\bar{z}_k}{z-z_k}(\eta(z-z_0))^{i\eta k(z_0)}exp(i\beta(z,z_0,\eta).
	\end{align*}
	Thus
	\begin{align}
		&|T(z,z_0,\eta)-T_0(z_0,\eta)(\eta(z-z_0))^{i\eta k(z_0)}|=\notag\\
		&|(\eta(z-z_0))^{i\eta k(z_0)}||\prod_{k\in \Delta^-_{z_0,\eta}}\dfrac{z-\bar{z}_k}{z-z_k}e^{i\beta(z,z_0,\eta)}-\prod_{k\in \Delta^-_{z_0,\eta}}\dfrac{z_0-\bar{z}_k}{z_0-z_k}e^{i\beta (z_0,z_0,\eta)}|.
	\end{align}
	We consider the case that $\eta=1$, the other case $\eta=-1$ is similarly. Note the fact that
	\begin{align}
	|(\eta(z-z_0))^{i k(z_0)}|&=|l^{i k(z_0)}exp\left( \frac{\psi}{2\pi}log(1+z_0|r(z_0|^2)\right)|\notag\\
	&\leq exp\left( \frac{1}{2}log(1+z_0|r(z_0|^2)\right)\notag\\
	&=\sqrt{1+z_0|r(z_0|^2},\label{z-z0}
	\end{align}
	and
	\begin{align}
	|\prod_{k\in \Delta^-_{z_0,\eta}}\dfrac{z-\bar{z}_k}{z-z_k}-\prod_{k\in \Delta^-_{z_0,\eta}}\dfrac{z_0-\bar{z}_k}{z_0-z_k}|=O(l),\hspace{0.5cm}\text{as }l\to 0,
	\end{align}
	with constants depending on $\psi$. And for function $\beta(z,z_0,1)$ we have
	\begin{align}
	&|\beta(z,z_0,1)-\beta(z_0,z_0,1)|=\notag\\
	&|-k(z_0)log(le^{i\psi}+1)+\int _{-\infty}^{z_0}\left(\dfrac{1}{s-z}-\dfrac{z}{s-z_0} \right) \left( k(s)-X(s)k(z_0)\right) ds|\notag\\
	&|\int _{-\infty}^{z_0-1}\left(\dfrac{1}{s-z}-\dfrac{z}{s-z_0} \right)k(s)ds+\int _{z_0-1}^{z_0}\left(\dfrac{1}{s-z}-\dfrac{z}{s-z_0} \right) \left( k(s)-X(s)k(z_0)\right) ds+O(t)|,
	\end{align}
	where
	\begin{align}
	\int _{-\infty}^{z_0-1}\left(\dfrac{1}{s-z}-\dfrac{z}{s-z_0} \right)k(s)ds=\int _{-\infty}^{z_0-1}\dfrac{le^{i\psi}}{(s-z)(s-z_0)} k(s)ds=O(l),
	\end{align}
	and
	\begin{align}
	&|\int _{z_0-1}^{z_0}\left(\dfrac{1}{s-z}-\dfrac{z}{s-z_0} \right) \left( k(s)-X(s)k(z_0)\right) ds|=\notag\\
	&|\int _{z_0-1}^{z_0}\dfrac{le^{i\psi}}{s-z_0+le^{i\psi}}k'(z')ds|\leq l\left( \int _{z_0-1}^{z_0}|\dfrac{1}{s-z_0+le^{i\psi}}|^2ds\right) ^{1/2}\parallel k'\parallel_{L^2}\notag\\
	&\leq Cl^{\frac{1}{2}},
	\end{align}	
	with C depending on $\parallel k\parallel_{H^2}$. Combining above equations we finally get the result.	
\end{proof}
We now use $T$ to define a new unknown function $M^{(1)}$:
\begin{equation}
M^{(1)}(z)=M(z)T(z)^{-\sigma_3}.\label{transm1}
\end{equation}
Note that if $M(z)$ is solution of \textbf{RHP 1}, then $M^{(1)}(z)$ satisfies the following \textbf{RHP}.\\

\textbf{Riemann-Hilbert Problem 2}. Find a matrix-valued function z$\in C \rightarrow M^{(1)}(z;x,t)$ which satisfies:

$\bullet$ Analyticity: $M^{(1)}(z;x,t)$ is meromorphic in $\mathbb{C}\setminus R$ and has single poles;

$\bullet$ Jump condition: $M^{(1)}$ has continuous boundary values $M^{(1)}_\pm$ on $R$ and
\begin{equation}
M^{(1)}_+(z;x,t)=M^{(1)}_-(z;x,t)V^{(1)}(z),\hspace{0.5cm}z \in R,\label{jump3}
\end{equation}
where
\begin{align}
&\text{as } z\in D^+, V^{(1)}(z)=
\left(\begin{array}{cc}
1 & -r(z)T(z)^2e^{-2it\theta}\\
0 & 1
\end{array}\right)\left(\begin{array}{cc}
1 & 0\\
-z\overline{r(z)}T(z)^{-2}e^{2it\theta} & 1
\end{array}\right), \\
&\text{as }z\in D^-, V^{(1)}(z)=\left(\begin{array}{cc}
1 & 0\\
-\dfrac{z\overline{r(z)}T_-(z)^{-2}}{1+z|r(z)|^2}e^{2it\theta} & 1
\end{array}\right)\left(\begin{array}{cc}
1 & -\dfrac{r(z)T_+(z)^{2}}{1+z|r(z)|^2}e^{-2it\theta}\\
0 & 1
\end{array}\right);
\end{align}

$\bullet$ Asymptotic behaviors:
\begin{align}
&M^{(1)}(z;x,t)= \left(\begin{array}{cc}
1 & 0\\
p & 1
\end{array}\right)+O(z^{-1}),\hspace{0.5cm}z \rightarrow \infty;\label{asymbehv4}
\end{align}

$\bullet$ Residue conditions: M has simple poles at each point in $ Z\bigcup \bar{Z}$ with:

\hspace{1cm} 1.   When $n\in \Delta^+_{z_0,\eta}$,
\begin{align}
&\res_{z=z_n}M^{(1)}(z)=\lim_{z\to z_n}M^{(1)}(z)\left(\begin{array}{cc}
0 & 0\\
c_nT(z_n)^{-2}e^{2i\theta_n t} & 0
\end{array}\right),\label{RES7}\\
&\res_{z=\bar{z}_n}M^{(1)}(z)=\lim_{z\to \bar{z}_n}M^{(1)}(z)\left(\begin{array}{cc}
0 & -z_n^{-1}T(\bar{z}_n)^2\bar{c}_ne^{-2i\bar{\theta}_n t}\\
0 & 0
\end{array}\right).\label{RES8}
\end{align}

\hspace{1cm} 2.   When $n\in \Delta^-_{z_0,\eta}$,
\begin{align}
&\res_{z=z_n}M^{(1)}(z)=\lim_{z\to z_n}M^{(1)}(z)\left(\begin{array}{cc}
0 & (c_n(1/T)'(z_n)^2e^{2i\theta_n t})^{-1}\\
0 & 0
\end{array}\right),\label{RES9}\\
&\res_{z=\bar{z}_n}M^{(1)}(z)=\lim_{z\to \bar{z}_n}M^{(1)}(z)\left(\begin{array}{cc}
0 & 0\\
(-z_nT'(\bar{z}_k)^2\bar{c}_ne^{-2i\bar{\theta}_n t})^{-1} & 0
\end{array}\right).\label{RES10}
\end{align}
\begin{proof}
The analyticity, jump condition and asymptotic behaviors of $M^{(1)}(z)$ is directly from its definition, the proposition \ref{proT} and the properties of $M$. As for residues, because $T(z)$ is analytic at each $z_n$ and $\bar{z}_n$ for $n\in \Delta^+_{z_0,\eta}$, from (\ref{RES4}), (\ref{RES5}) and (\ref{transm1}) we obtain residue conditions at these point immediately. For $n\in \Delta^-_{z_0,\eta}$, we denote $M(z)=\left(M_1(z), M_2(z) \right) $, then $M^{(1)}(z)=\left(M^{(1)}_1(z), M^{(1)}_2(z) \right) $ =  $\left(M_1(z)T^{-1}(z), M_2(z)T(z) \right) $. $T(z)$ has a simple zero at $\bar{z}_n$ and a pole at $z_n$, so $z_n$ is no longer the pole of $M^{(1)}_1(z)$ with $\bar{z}_n$ becoming the pole of it. And $M^{(1)}_2(z)$ has opposite situation. It has pole at $z_n$ and a removable singularity at $\bar{z}_n$. We consider the residue condition of $M^{(1)}_2(z)$  at pole $z_n$,
\begin{align}
M^{(1)}_1(z_n)&=\lim_{z\to z_n}M_1(z)T^{-1}(z)=\res_{z=z_n}M_1(z)(1/T)'(z_n)\nonumber\\
&=c_n(1/T)'(z_n)e^{2i\theta_n t}M_2(z_n),\nonumber\\
\res_{z=z_n}M^{(1)}_2(z)&=\res_{z=z_n}M_2(z)T(z)=M_2(z_k)\left[ (1/T)'(z_k)\right] ^{-1}\nonumber\\
&=(c_n(1/T)'(z_n)^2e^{2i\theta_n t})^{-1}M^{(1)}_1(z_n).
\end{align}
Then we have (\ref{RES9}), and the (\ref{RES10}) is similarly.
\end{proof}

\section{Mixed $\bar{\partial}$-Riemann-Hilbert Problem }

\quad Next we introduce a transformations of the jump matrix which deform the contours to another contours defined as follow:
\begin{align}
&\Sigma_k=z_0+e^{(2k-1)i\pi/4}R_+,\hspace{0.5cm}k=1,2,3,4,\\
&\Sigma^{(2)}=\Sigma_1\cup\Sigma_2\cup\Sigma_3\cup\Sigma_4.
\end{align}
Along it the jumps are decaying. But the price we pay for this non-analytic transformation is that it appear new unknown nonzero $\bar{\partial}$ derivatives inside the regions in which the extensions are introduced and satisfies a mixed $\bar{\partial}$-Riemann-Hilbert problem.
$R$ and $\Sigma^{(2)}$ separate $C$ to six open sectors. We denote these sectors by $\Omega_k$, $k=1,...,6$,  starting with sector $\Omega_1$ between $I^\eta_{z_0}$ and $\Sigma_1$ and numbered consecutively continuing counterclockwise for $\eta$ = 1 ($\eta$ = -1 is similarly) as shown in Figure \ref{figR2}.
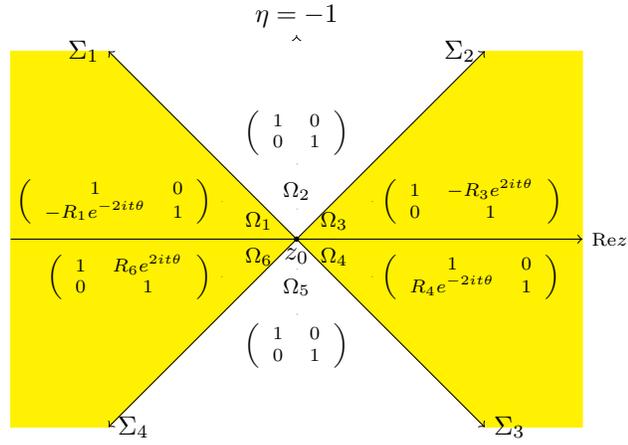
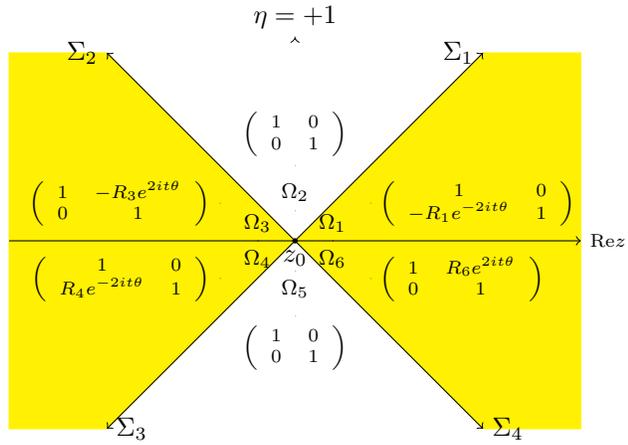
\begin{figure}
	\centering
	\subfigure[]{
		\begin{tikzpicture}[node distance=2cm]
		\draw[yellow, fill=yellow] (0,0)--(-2.5,2.5)--(-3.8,2.5)--(-3.8,-2.5)--(-2.5,-2.5)--(0,0);
		\draw[yellow, fill=yellow] (0,0)--(2.5,2.5)--(3.8,2.5)--(3.8,-2.5)--(2.5,-2.5)--(0,0);
		\draw[->](0,2.7)--(0,2.7)node[above]{$\eta=-1$};
		\draw[->](0,0)--(2.5,2.5)node[left]{$\Sigma_2$};
		\draw[->](0,0)--(-2.5,2.5)node[left]{$\Sigma_1$};
		\draw[->](0,0)--(-2.5,-2.5)node[right]{$\Sigma_4$};
		\draw[->](0,0)--(2.5,-2.5)node[right]{$\Sigma_3$};
		\draw[->](-3.8,0)--(3.8,0)node[right]{\scriptsize Re$z$};
		\coordinate (A) at (1,0.5);
		\coordinate (B) at (1,-0.5);
		\coordinate (G) at (-1,0.5);
		\coordinate (H) at (-1,-0.5);
		\coordinate (I) at (0,0);
		\coordinate (C) at (-0.5,0);
		\fill (C) circle (0pt) node[above] {\footnotesize$\Omega_1$};
		\coordinate (E) at (0,0.4);
		\fill (E) circle (0pt) node[above] {\footnotesize$\Omega_2$};
		\coordinate (D) at (0.5,0);
		\fill (D) circle (0pt) node[above] {\footnotesize$\Omega_3$};
		\coordinate (F) at (0.5,0);
		\fill (F) circle (0pt) node[below] {\footnotesize$\Omega_4$};
		\coordinate (J) at (0,-0.4);
		\fill (J) circle (0pt) node[below] {\footnotesize$\Omega_5$};
		\coordinate (k) at (-0.5,0);
		\fill (k) circle (0pt) node[below] {\footnotesize$\Omega_6$};
		\fill (A) circle (0pt) node[right] {\scriptsize$\left(\begin{array}{cc}
			1 & -R_3e^{2it\theta}\\
			0 & 1
			\end{array}\right)$};
		\fill (B) circle (0pt) node[right] {\scriptsize$\left(\begin{array}{cc}
			1 & 0\\
			R_4e^{-2it\theta} & 1
			\end{array}\right)$};
		\fill (G) circle (0pt) node[left] {\scriptsize$\left(\begin{array}{cc}
			1&0\\
			-R_1e^{-2it\theta}&1
			\end{array}\right)$};
		\fill (H) circle (0pt) node[left] {\scriptsize$\left(\begin{array}{cc}
			1 & R_6e^{2it\theta}\\
			0 & 1
			\end{array}\right)$};
		\fill (I) circle (1pt) node[below] {$z_0$};
		\coordinate (L) at (0,1);
		\coordinate (M) at (0,-1);
		\fill (L) circle (0pt) node[above] {\scriptsize$\left(\begin{array}{cc}
			1 & 0\\
			0 & 1
			\end{array}\right)$};
		\fill (M) circle (0pt) node[below] {\scriptsize$\left(\begin{array}{cc}
			1 & 0\\
			0 & 1
			\end{array}\right)$};
		\end{tikzpicture}
	}
	\subfigure[]{
		\begin{tikzpicture}[node distance=2cm]
		\draw[yellow, fill=yellow] (0,0)--(-2.5,2.5)--(-3.8,2.5)--(-3.8,-2.5)--(-2.5,-2.5)--(0,0);
		\draw[yellow, fill=yellow] (0,0)--(2.5,2.5)--(3.8,2.5)--(3.8,-2.5)--(2.5,-2.5)--(0,0);
		\draw[->](0,2.7)--(0,2.7)node[above]{$\eta=+1$};
		\draw[->](0,0)--(2.5,2.5)node[left]{$\Sigma_1$};
		\draw[->](0,0)--(-2.5,2.5)node[left]{$\Sigma_2$};
		\draw[->](0,0)--(-2.5,-2.5)node[right]{$\Sigma_3$};
		\draw[->](0,0)--(2.5,-2.5)node[right]{$\Sigma_4$};
		\draw[->](-3.8,0)--(3.8,0)node[right]{\scriptsize Re$z$};
		\coordinate (A) at (1,0.5);
		\coordinate (B) at (1,-0.5);
		\coordinate (G) at (-1,0.5);
		\coordinate (H) at (-1,-0.5);
		\coordinate (I) at (0,0);
		\coordinate (C) at (-0.5,0);
		\fill (C) circle (0pt) node[above] {\footnotesize$\Omega_3$};
		\coordinate (E) at (0,0.4);
		\fill (E) circle (0pt) node[above] {\footnotesize$\Omega_2$};
		\coordinate (D) at (0.5,0);
		\fill (D) circle (0pt) node[above] {\footnotesize$\Omega_1$};
		\coordinate (F) at (0.5,0);
		\fill (F) circle (0pt) node[below] {\footnotesize$\Omega_6$};
		\coordinate (J) at (0,-0.4);
		\fill (J) circle (0pt) node[below] {\footnotesize$\Omega_5$};
		\coordinate (k) at (-0.5,0);
		\fill (k) circle (0pt) node[below] {\footnotesize$\Omega_4$};
		\fill (A) circle (0pt) node[right] {\scriptsize$\left(\begin{array}{cc}
			1 & 0\\
			-R_1e^{-2it\theta} & 1
			\end{array}\right)$};
		\fill (B) circle (0pt) node[right] {\scriptsize$\left(\begin{array}{cc}
			1 & R_6e^{2it\theta}\\
			0 & 1
			\end{array}\right)$};
		\fill (G) circle (0pt) node[left] {\scriptsize$\left(\begin{array}{cc}
			1& -R_3e^{2it\theta}\\
			0&1
			\end{array}\right)$};
		\fill (H) circle (0pt) node[left] {\scriptsize$\left(\begin{array}{cc}
			1 & 0\\
			R_4e^{-2it\theta} & 1
			\end{array}\right)$};
		\fill (I) circle (1pt) node[below] {$z_0$};
		\coordinate (L) at (0,1);
		\coordinate (M) at (0,-1);
		\fill (L) circle (0pt) node[above] {\scriptsize$\left(\begin{array}{cc}
			1 & 0\\
			0 & 1
			\end{array}\right)$};
		\fill (M) circle (0pt) node[below] {\scriptsize$\left(\begin{array}{cc}
			1 & 0\\
			0 & 1
			\end{array}\right)$};
		\end{tikzpicture}
	}
	\caption{In the yellow region, $|e^{2it\theta}|\to \infty$ when $t\to\pm\infty$ respectively. And in white region, $|e^{2it\theta}|\to 0$ when $t\to\pm\infty$ respectively}
	\label{figR2}
\end{figure}

Additionally, let
\begin{equation}
\mu=\frac{1}{2}\min_{\lambda\neq\gamma\in Z \cup \bar{Z}}|\lambda -\gamma|.
\end{equation}
As we assuming there is no pole on the real axis, we obtain $\mu<$dist($Z,R$). Then we define $X_Z \in C_0^\infty(C,[0,1])$ which only supported on the
neighborhood of $Z\cup\bar{Z}$,
\begin{equation}
X_Z(z)=\Bigg\{\begin{array}{ll}
1 &\text{dist}(z,Z\cup\bar{Z})<\mu/3\\
0 &\text{dist}(z,Z\cup\bar{Z})>2\mu/3.\\
\end{array}
\end{equation}
In order to deform the contour $R$ to the contour $\Sigma^{(2)}$, we introduce a new unknown function $M^{(2)}$.
\begin{equation}
M^{(2)}(z)=M^{(1)}(z)R^{(2)}(z).\label{transm2}
\end{equation}
And we need $R$ satisfy the following conditions: First, $M^{(2)}$ has no jump on the real axis, so we choose the  the boundary values of $R^{(2)}(z)$  through the factorization of $V^{(1)}(z)$ in (\ref{jump3}) where the the new jumps on $\Sigma^{(2)}$ match a well known model RHP; Second, we need to control the norm of $R^{(2)}(z)$, so that the $\bar{\partial}$-contribution to the long-time asymptotics of $q(x,t)$ can be ignored; Third the residues are unaffected by the transformation.  So we choose $R^{(2)}(z)$ as shown in Figure \ref{figR2}, where the function $R_j$, $j=1,3,4,6$, is defined as follow.
\begin{proposition}\label{proR}
	 $R_j$: $\bar{\Omega}_j\to C$, $j=1,3,4,6$ have boundary values as follow:
\begin{align}
&R_1(z)=\Bigg\{\begin{array}{ll}
z\bar{r}(z)T(z)^{-2} & z\in I_{z_0}^\eta\\
z_0\bar{r}(z_0)T(z_0)^{-2}(\eta(z-z_0))^{-2i\eta k(z_0)}(1-X_Z(z))  &z\in \Sigma_1\\
\end{array},\\
&R_3(z)=\Bigg\{\begin{array}{ll}
\dfrac{r(z_0)T(z_0)^2}{1+z_0|r(z_0)|^2}(\eta(z-z_0))^{2i\eta k(z_0)}(1-X_Z(z))  &z\in \Sigma_2\\
\dfrac{r(z)T_+(z)^2}{1+z|r(z)|^2} & z\in R\setminus (I_{z_0}^\eta\cup\left\lbrace z_0 \right\rbrace )\\
\end{array},\\
&R_4(z)=\Bigg\{\begin{array}{ll}
-\dfrac{zr(z)T_-(z)^{-2}}{1+z|r(z)|^2} & z\in R\setminus (I_{z_0}^\eta\cup\left\lbrace z_0 \right\rbrace )\\
-\dfrac{z_0r(z_0)T(z_0)^{-2}}{1+z_0|r(z_0)|^2}(\eta(z-z_0))^{-2i\eta k(z_0)}(1-X_Z(z))  &z\in \Sigma_3\\
\end{array},\\
&R_6(z)=\Bigg\{\begin{array}{ll}
-r(z_0)T(z_0)^{2}(\eta(z-z_0))^{2i\eta k(z_0)}(1-X_Z(z))  &z\in \Sigma_4\\
-r(z)T(z)^{2} & z\in I_{z_0}^\eta\\
\end{array}.
\end{align}	
then we can find a fixed constant $c_1$=$c_1(q_0)$, have
\begin{align}
&|R_j(z)|\leq c_1\left( sin^2(arg(z-z_0))+\langle Re(z)\rangle^{-1/2}\right),\label{R}\\
&| \bar{\partial}R_j(z)|\leq c_1\left(|\bar{\partial}X_Z(z)|+|z^{m_k}p_k'(Rez)|+|z-z_0|^{-1/2}\right),\label{dbarRj}\\
&\bar{\partial}R_j(z)=0,\hspace{0.5cm}\text{if } z\in \Omega_2\cup\Omega_5\text{ or }\text{dist}(z,Z\cup\bar{Z})<\mu/3,
\end{align}
where $m_1=m_4=1$, $m_3=m_6=0$, and
\begin{align}
&p_1(z)=\bar{r}(z),\hspace{3cm}p_3(z)=\dfrac{r(z)}{1+z|r(z)|^2},\\
&p_4(z)=-\dfrac{zr(z)}{1+z|r(z)|^2},\hspace{1.5cm}p_6(z)=-r(z).
\end{align}
\end{proposition}
\begin{proof}
	Define the functions
	\begin{align}
	&f_1(z)=z_0p_1(z_0)T(z)^2T(z_0)^{-2}(\eta(z-z_0))^{-2i\eta k(z_0)},\hspace{1cm}z\in \bar{\Omega}_1\\
	&f_3(z)=p_3(z_0)T(z)^{-2}T(z_0)^2(\eta(z-z_0))^{2i\eta k(z_0)},\hspace{1.45cm}z\in \bar{\Omega}_3
	\end{align}
	Let $z-z_0=le^{i\psi}$, then using the above function we can give the construction  for $R_1(z)$ and $R_3(z)$
	\begin{align}
	&R_1(z)=\left[ f_1(z)+\eta(Rezp_1(Rez)-f_1(z))cos(2\psi)\right]T(z)^{-2}(1-X_Z(z)) ,\\
	&R_3(z)=\left[ f_3(z)+\eta(p_3(Rez)-f_3(z))cos(2\psi)\right]T(z)^{2}(1-X_Z(z)).
	\end{align}
	The construction of $R_4(z)$ and $R_6(z)$ are similarly. We proof this proposition for example $R_1(z)$ , the case of others is same. And from the definition of $R_1(z)$ we can easily obtain that $\bar{\partial}R_j(z)=0$, $\text{if } z\in \Omega_2\cup\Omega_5\text{ or }\text{dist}(z,Z\cup\bar{Z})<\mu/3$. Then we estimate $|R_1(z)|$,
	\begin{align}
	|R_1(z)|\leq&|f_1(z)T(z)^{-2}(1-X_Z(z))(1-\eta cos(2\psi))|\nonumber\\
	&+|Rezp_1(Rez)T(z)^{-2}(1-X_Z(z))cos(2\psi)|.
	\end{align}
	Note that $(1-X_Z(z))$ is zero in $(z,Z\cup\bar{Z})<\mu/3$, so $|T(z)^{-2}(1-X_Z(z))|$ is bounded. And from $r\in H^2(R)$, which means $p_1\in H^2(R)$ we have $|p_1(u)|\lesssim \langle u\rangle^{-1/2}$. Together with (\ref{z-z0}) we have (\ref{R}). Since
	\begin{equation*}
	\bar{\partial}=\frac{1}{2}\left( \partial_x+i\partial_y\right) =\frac{e^{i\psi}}{2}\left( \partial_l+il^{-1}\partial_\psi\right),
	\end{equation*}
	we have
	\begin{align}
	\bar{\partial}R_1(z)=&-\left[ f_1(z)+\eta(Rezp_1(Rez)-f_1(z))cos(2\psi)\right]T(z)^{-2}\bar{\partial}X_Z(z)\nonumber\\
	&+\eta\frac{Rez}{2}p_1'(Rez)cos(2\psi)T(z)^{-2}(1-X_Z(z))\nonumber\\
	&-\dfrac{ie^{i\psi}\eta}{|z-z_0|}(Rezp_1(Rez)-f_1(z))sin(2\psi)T(z)^{-2}(1-X_Z(z)).
	\end{align}
	Because $\bar{\partial}X_Z(z)$ is supported in dist$(z,Z\cup\bar{Z})\in [\mu/3,2\mu/3]$, similarly we have the first two term are bounded. For the last term, we have
	\begin{align}
	|Rezp_1(Rez)-f_1(z)|\leq&|Rezp_1(Rez)-z_0p_1(z_0)|\nonumber\\
	&+|z_0p_1(z_0)||1-T(z)^2T(z_0)^{-2}(\eta(z-z_0))^{-2i\eta k(z_0)}|.
	\end{align}
	The first term using Cauchy-Schwarz we have
	\begin{align}
	|Rezp_1(Rez)-z_0p_1(z_0)|&\leq|\int_{z_0}^{Rez}sp'(s)+p(s)ds|\nonumber\\
	&\leq|z-z_0|^{1/2}\parallel sp'(s)+p(s)\parallel_{L^2(R)}\nonumber\\
	&\leq|z-z_0|^{1/2}2\parallel p(s)\parallel_{H^2(R)}.
	\end{align}
	And together with \textbf{Proposition \ref{proT}}, which also imply $|T(z_0)|$ and $|(\eta(z-z_0))^{-2i\eta k(z_0)}|$ are bounded in a neighborhood of $z_0$, we come to the consequence for the estimation of $|\bar{\partial}R_1|$.
\end{proof}
 Then we use $R^{(2)}$ to define a new unknown function
 \begin{equation}
 M^{(2)}=M^{(1)}R^{(2)},
 \end{equation}
which satisfies the following mixed $\bar{\partial}$-Riemann-Hilbert problem.

\textbf{mixed $\bar{\partial}$-Riemann-Hilbert problem 3}. Find a matrix-valued function z$\in C \rightarrow M^{(2)}(z;x,t)$ with following properties:

$\bullet$ Analyticity: $M^{(2)}(z;x,t)$ is continuous with sectionally continuous first partial derivatives in  $\mathbb{C}\setminus (\Sigma^{(2)}\cup Z\cup \bar{Z})$ and meromorphic in $\Omega_2\cup\Omega_5$;

$\bullet$ Jump condition: $M^{(2)}$ has continuous boundary values $M^{(2)}_\pm$ on $R$ and
\begin{equation}
M^{(2)}_+(z;x,t)=M^{(2)}_-(z;x,t)V^{(2)}(z),\hspace{0.5cm}z \in R,\label{jump4}
\end{equation}
where
\begin{equation}
V^{(2)}(z)=I+(1-X_Z(z))V_0^{(2)},
\end{equation}
\begin{equation}
V_0^{(2)}=\left\{\begin{array}{llll}
\left(\begin{array}{cc}
0 & 0\\
-z\overline{r(z)}T(z_0)^{-2}(\eta(z-z_0))^{-2i\eta k(z_0)}e^{2it\theta} & 0
\end{array}\right), & z\in \Sigma_1;\\
\\
\left(\begin{array}{cc}
0 & -\dfrac{r(z)T(z_0)^{2}}{1+z|r(z)|^2}(\eta(z-z_0))^{2i\eta k(z_0)}e^{-2it\theta}\\
0 & 0
\end{array}\right),  &z\in \Sigma_2;\\
\\
\left(\begin{array}{cc}
0 & 0\\
-\dfrac{z\overline{r(z)}T(z_0)^{-2}}{1+z|r(z)|^2}(\eta(z-z_0))^{-2i\eta k(z_0)}e^{2it\theta} & 0
\end{array}\right),  &z\in \Sigma_3;\\
\\
\left(\begin{array}{cc}
0 & -r(z)T(z_0)^2(\eta(z-z_0))^{2i\eta k(z_0)}e^{-2it\theta}\\
0 & 0
\end{array}\right),  &z\in \Sigma_4,\\
\end{array}\right.
\end{equation}

$\bullet$ Asymptotic behaviors:
\begin{align}
&M^{(2)}(z;x,t) \sim \left(\begin{array}{cc}
1 & 0
\end{array}\right)+O(z^{-1}),\hspace{0.5cm}z \rightarrow \infty.\label{asymbehv5}
\end{align}

$\bullet$ $\bar{\partial}$-Derivative: For $\mathbb{C}\setminus (\Sigma^{(2)}\cup Z\cup \bar{Z})$ we have $\bar{\partial}M^{(2)}=M^{(1)}\bar{\partial}R^{(2)}$, where
\begin{equation}
\bar{\partial}R^{(2)}=\left\{\begin{array}{llll}
\left(\begin{array}{cc}
0 & 0\\
-\bar{\partial}R_1e^{2it\theta} & 0
\end{array}\right), & z\in \Omega_1;\\ \\
\left(\begin{array}{cc}
0 & -\bar{\partial}R_3e^{-2it\theta}\\
0 & 0
\end{array}\right),  &z\in \Omega_3;\\ \\
\left(\begin{array}{cc}
0 & 0\\
\bar{\partial}R_4e^{2it\theta} & 0
\end{array}\right),  &z\in \Omega_4;\\ \\
\left(\begin{array}{cc}
0 & \bar{\partial}R_6e^{-2it\theta}\\
0 & 0
\end{array}\right),  &z\in \Omega_6;\\ \\
0,  &z\in \Omega_2\cup\Omega_5.\\
\end{array}\right. \label{DBARR2}
\end{equation}

$\bullet$ Residue conditions: M has simple poles at each point in $ Z\bigcup \bar{Z}$ with:

\hspace{1cm} 1.   When $n\in \Delta^+_{z_0,\eta}$,
\begin{align}
&\res_{z=z_n}M^{(2)}(z)=\lim_{z\to z_n}M^{(2)}(z)\left(\begin{array}{cc}
0 & 0\\
c_nT(z_n)^{-2}e^{2i\theta_n t} & 0
\end{array}\right),\\
&\res_{z=\bar{z}_n}M^{(2)}(z)=\lim_{z\to \bar{z}_n}M^{(2)}(z)\left(\begin{array}{cc}
0 & -z_n^{-1}T(\bar{z}_n)^2\bar{c}_ne^{-2i\bar{\theta}_n t}\\
0 & 0
\end{array}\right).
\end{align}

\hspace{1cm} 2.   When $n\in \Delta^-_{z_0,\eta}$,
\begin{align}
&\res_{z=z_n}M^{(2)}(z)=\lim_{z\to z_n}M^{(2)}(z)\left(\begin{array}{cc}
0 & (c_n(1/T)'(z_n)^2e^{2i\theta_n t})^{-1}\\
0 & 0
\end{array}\right),\\
&\res_{z=\bar{z}_n}M^{(2)}(z)=\lim_{z\to \bar{z}_n}M^{(2)}(z)\left(\begin{array}{cc}
0 & 0\\
(-z_nT'(\bar{z}_k)^2\bar{c}_ne^{-2i\bar{\theta}_n t})^{-1} & 0
\end{array}\right).
\end{align}

\section{  Model RH problem and  pure $\bar{\partial}$-problem }
\quad To solve \textbf{RHP 3}, we decompose it into a model RH  Problem and a pure $\bar{\partial}$-Problem.
First we build a solution $M^{RHP}$ to a model RH problem  as following

\textbf{Riemann-Hilbert problem 4}. Find a matrix-valued function  $  M^{RHP}(z;x,t)$ with following properties:

$\bullet$ Analyticity: $M^{RHP}(z;x,t)$ is meromorphic in $\mathbb{C}\setminus (\Sigma^{(2)}\cup Z\cup \bar{Z})$;

$\bullet$ Jump condition: $M^{RHP}$ has continuous boundary values $M^{RHP}_\pm$ on $R$ and
\begin{equation}
M^{RHP}_+(z;x,t)=M^{RHP}_-(z;x,t)V^{(2)}(z),\hspace{0.5cm}z \in R;\label{jump5}
\end{equation}

$\bullet$ Symmetry: $M^{RHP}_{22}(z)=\overline{M^{RHP}_{11}(\bar{z})}$, $M^{RHP}_{21}(z)=-z\overline{M^{RHP}_{12}(\bar{z})}$;

$\bullet$ Asymptotic behaviors:
\begin{align}
&M^{RHP}(z;x,t) \sim \left(\begin{array}{cc}
1 & 0\\
\alpha & 0
\end{array}\right)+O(z^{-1}),\hspace{0.5cm}z \rightarrow \infty\label{asymbehv6}
\end{align}
 for a constant $\alpha$ determined by the symmetry condition above;

$\bullet$ Residue conditions: $M^{RHP}$ has simple poles at each point in $ Z\bigcup \bar{Z}$ with:

\hspace{1cm} 1.   When $n\in \Delta^+_{z_0,\eta}$,
\begin{align}
&\res_{z=z_n}M^{RHP}(z)=\lim_{z\to z_n}M^{RHP}(z)\left(\begin{array}{cc}
0 & 0\\
c_nT(z_n)^{-2}e^{2i\theta_n t} & 0
\end{array}\right),\label{RES11}\\
&\res_{z=\bar{z}_n}M^{RHP}(z)=\lim_{z\to \bar{z}_n}M^{RHP}(z)\left(\begin{array}{cc}
0 & -z_n^{-1}T(\bar{z}_n)^2\bar{c}_ne^{-2i\bar{\theta}_n t}\\
0 & 0
\end{array}\right).\label{RES12}
\end{align}

\hspace{1cm} 2.   When $n\in \Delta^-_{z_0,\eta}$,
\begin{align}
&\res_{z=z_n}M^{RHP}(z)=\lim_{z\to z_n}M^{RHP}(z)\left(\begin{array}{cc}
0 & (c_n(1/T)'(z_n)^2e^{2i\theta_n t})^{-1}\\
0 & 0
\end{array}\right),\label{RES13}\\
&\res_{z=\bar{z}_n}M^{RHP}(z)=\lim_{z\to \bar{z}_n}M^{RHP}(z)\left(\begin{array}{cc}
0 & 0\\
(-z_nT'(\bar{z}_k)^2\bar{c}_ne^{-2i\bar{\theta}_n t})^{-1} & 0
\end{array}\right).\label{RES14}
\end{align}

And we will proof the existence of $M^{RHP}$ and construct its asymptotic expansion for $t\to\infty$ later. Before it we first consider to use $M^{RHP}$ to get the solution of following pure $\bar{\partial}$-problem
\begin{equation}
M^{(3)}(z)=M^{(2)}(z)M^{RHP}(z)^{-1}.\label{transm3}
\end{equation}

\textbf{$\bar{\partial}$-problem 5}. Find a matrix-valued function z$\in C \rightarrow M^{(2)}(z;x,t)$ with following properties:

$\bullet$ Analyticity: $M^{(3)}(z;x,t)$ is continuous with sectionally continuous first partial derivatives in  $\mathbb{C}\setminus (\Sigma^{(2)}\cup Z\cup \bar{Z})$ and meromorphic in $\Omega_2\cup\Omega_5$.

$\bullet$ Asymptotic behavior:
\begin{align}
&M^{(3)}(z;x,t) \sim \left(\begin{array}{cc}
1 & 0
\end{array}\right)+O(z^{-1}),\hspace{0.5cm}z \rightarrow \infty;\label{asymbehv7}
\end{align}

$\bullet$ $\bar{\partial}$-Derivative: For $\mathbb{C}\setminus (\Sigma^{(2)}\cup Z\cup \bar{Z})$ we have $\bar{\partial}M^{(3)}=M^{(3)}W^{(3)}$, where
\begin{equation}
W^{(3)}=M^{RHP}(z)\bar{\partial}R^{(2)}(z)M^{RHP}(z)^{-1.}
\end{equation}
\begin{proof}
	For the property of solutions $M^{(2)}$ and $M^{RHP}$ of \textbf{mixed-$\bar{\partial}$-RHP} and \textbf{RHP 4} respectively, the analyticity and asymptotic behavior come immediately. Since $M^{(2)}$ and $M^{RHP}$ have same jump matrix, we have
	\begin{align*}
	M_-^{(3)}(z)^{-1}M_+^{(3)}(z)&=M_-^{RHP}(z)M_-^{(2)}(z)^{-1}M_+^{(2)}(z)M_+^{RHP}(z)^{-1}\\
	&=M_-^{RHP}(z)V^{(2)}(z)\left( M_-^{RHP}(z)V^{(2)}(z)\right) ^{-1}\\
	&=I,
	\end{align*}
	which means $ M^{(3)}$ has no jumps and is everywhere continuous. Then we proof that $ M^{(3)}$ has no pole. For instance, if $\lambda \in Z\cup \bar{Z}$ and let $N$ denote the  nilpotent matrix which appears in the left side of the corresponding residue condition of \textbf{mixed-$\bar{\partial}$-RHP} and \textbf{RHP 4},  we have the Laurent expansions in $z-\lambda$
	\begin{equation}
	M^{(2)}(z)=a(\lambda) \left[ \dfrac{N}{z-\lambda}+I\right] +O(z-\lambda),\hspace{0.5cm}	M^{RHP}(z)=A(\lambda) \left[ \dfrac{N}{z-\lambda}+I\right] +O(z-\lambda),
	\end{equation}
where $a(\lambda)$ and $A(\lambda)$ are the constant row vector and matrix in their respective expansions. And from $M^{RHP}(z)^{-1}=\sigma_2M^{RHP}(z)^T\sigma_2$, we have
\begin{align}
M^{(3)}(z)&=\left\lbrace a(\lambda) \left[ \dfrac{N}{z-\lambda}+I\right]\right\rbrace \left\lbrace\left[ \dfrac{-N}{z-\lambda}+I\right]\sigma_2A(\lambda)^T\sigma_2\right\rbrace + O(z-\lambda)\nonumber\\
&=O(1),
\end{align}
which means  $ M^{(3)}$ has removable singularities at $\lambda$. And the $\bar{\partial}$-derivative of  $ M^{(3)}$ is from the $\bar{\partial}$-derivative of  $ M^{(3)}$ and the analyticity of $M^{RHP}$.
\end{proof}
Then we begin to prove the existence of  $M^{RHP}$ and explain its characteristics. We construct the solution $M^{RHP}$ as follow:
\begin{equation}
M^{RHP}=\left\{\begin{array}{ll}
E(z)M^{(out)} & z\notin U_{z_0}\\
E(z)M^{(z_0)}  &z\in U_{z_0}\\
\end{array}\right.,\label{transm4}
\end{equation}
where $ U_{z_0}$ is the neighborhood of $z_0$
\begin{equation}
U_{z_0}=\left\lbrace z:|z-z_0|\leq \mu/3\right\rbrace .
\end{equation}
From the definition we can easily find that $M^{RHP}$ is pole free. This decomposition decompose $M^{RHP}$ to two part: $M^{(out)}$ solves the pure RHP obtained by ignoring the jump conditions of \textbf{RHP 4} which is shown in Section 6; $M^{(z_0)}$ uses parabolic cylinder functions to build a matrix whose jumps exactly match those of $M^{(2)}$ in a neighborhood of the critical point $z_0$ which is shown in Section 7. And $E(z)$ is the error function, which is a solution of a small-norm Riemann-Hilbert problem shown in Section 9.

\section{The outer RH  model problem}
\quad In this section we build outer model Riemann-Hilbert problem and research its property. From  the previous section we find that the matrix function $M^{RHP}$ is meromorphic away from the contour $\Sigma^{(2)}$, and on the $\Sigma^{(2)}$, the boundry value satisfy $M^{RHP}_+(z;x,t)=M^{RHP}_-(z;x,t)V^{(2)}(z)$ which have following proposition.
\begin{proposition}\label{pro3v2}
	For the jump matrix of $M^{RHP}$, we have
	\begin{equation}
	\parallel V^{(2)}-I\parallel_{L^\infty(\Sigma^{(2)})}=O(e^{-4t|z-z_0|^2}).
	\end{equation}
\end{proposition}
\begin{proof}
	We proof it for example in $\Sigma_1$. The others case can be proof in the same way. For $z\in\Sigma_1$, we have
	\begin{equation}
	\parallel V^{(2)}-I\parallel_{L^\infty(\Sigma^{(2)})}=\parallel R_1(z)e^{2it\theta}\parallel_{L^\infty(\Sigma^{(2)})}.
	\end{equation}
	And from \textbf{Proposition \ref{proT}}, \textbf{Proposition \ref{proR}} and (\ref{theta}) we have
	\begin{equation}
	|R_1(z)e^{2it\theta}|\lesssim e^{-Im(2t\theta)}+\langle Re(z)\rangle^{-1/2}e^{-Im(2t\theta)}\lesssim \left( 1+\langle Re(z)\rangle^{-1/2}\right) e^{-4t|z-z_0|^2}.
	\end{equation}
	So we come to the consequence.
\end{proof}
This proposition means that the jump $V^{(2)}$ is uniformly near identity. So outside the $U_{z_0}$ there is only exponentially small error (in t) by completely ignoring the jump condition of $M^{RHP}$. Then we can introduce following outer model problem.

\textbf{Riemann-Hilbert problem 6}. Find a matrix-valued function z$\in C \rightarrow M^{(out)}(z;x,t)$ with following properties:

$\bullet$ Analyticity: $M^{(out)}(z;x,t)$ is meromorphic in $\mathbb{C}\setminus (\Sigma^{(2)}\cup Z\cup \bar{Z})$;

$\bullet$ Symmetry: $M^{(out)}_{22}(z)=\overline{M^{(out)}_{11}(\bar{z})}$, $M^{(out)}_{21}(z)=-z\overline{M^{(out)}_{12}(\bar{z})}$;

$\bullet$ Asymptotic behaviors:
\begin{align}
&M^{(out)}(z;x,t) \sim \left(\begin{array}{cc}
1 & 0\\
\alpha & 0
\end{array}\right)+O(z^{-1}),\hspace{0.5cm}z \rightarrow \infty
\end{align}
for a constant $\alpha$ determined by the symmetry condition above;

$\bullet$ Residue conditions: $M^{(out)}$ has simple poles at each point in $ Z\bigcup \bar{Z}$ satisfying the same residue relations of $M^{RHP}$.

\begin{proposition}\label{unim}	
the Riemann-Hilbert problem as follow, which is the reflectionless case $r\equiv0$ of \textbf{RHP MNLS}, exist unique solution.

\textbf{problem 1}. Given discrete data $\sigma_d=\left\lbrace(z_k,c_k) \right\rbrace _{k=1}^N$, and $Z=${$z_k$}.Find a matrix-valued function z$\in C \rightarrow m(z;x,t|\sigma_d)$ with following properties:

$\bullet$ Analyticity: $m(z;x,t|\sigma_d)$ is meromorphic in $\mathbb{C}\setminus (\Sigma^{(2)}\cup Z\cup \bar{Z})$;

$\bullet$ Symmetry: $m_{22}(z;x,t|\sigma_d)=\overline{m_{11}(\bar{z};x,t|\sigma_d)}$, $m_{21}(z;x,t|\sigma_d)=-z\overline{m_{12}(\bar{z};x,t|\sigma_d)}$, which means $m(z;x,t|\sigma_d)=z^{\sigma_3/2}\sigma_2\overline{m(\bar{z};x,t|\sigma_d)}\sigma_2^{-1}z^{-\sigma_3/2}$;

$\bullet$ Asymptotic behaviors:
\begin{align}
&m(z;x,t|\sigma_d) \sim \left(\begin{array}{cc}
1 & 0\\
\alpha & 1
\end{array}\right)+O(z^{-1}),\hspace{0.5cm}z \rightarrow \infty
\end{align}
for a constant $\alpha$ determined by the symmetry condition above;

$\bullet$ Residue conditions: $m(z;x,t|\sigma_d)$ has simple poles at each point in $ Z\bigcup \bar{Z}$ satisfying
\begin{align}
&\res_{z=z_n}m(z;x,t|\sigma_d)=\lim_{z\to z_n}m(z;x,t|\sigma_d)\tau_k,\\
&\res_{z=\bar{z}_n}m(z;x,t|\sigma_d)=\lim_{z\to \bar{z}_n}m(z;x,t|\sigma_d)\hat{\tau}_k,
\end{align}
where $\tau_k$ is a nilpotent matrix satisfies
\begin{align}
\tau_k=\left(\begin{array}{cc}
0 & 0\\
\gamma_k & 0
\end{array}\right),\hspace{0.5cm}\hat{\tau}_k=z_k^{\sigma_3/2}\sigma_2\overline{\tau_k}\sigma_2^{-1}z_k^{-\sigma_3/2},\hspace{0.5cm} \gamma_k=c_ke^{2i[(z-\rho)x+2(z-\rho)^2t]}.
\end{align}
Moreover, the solution satisfies
\begin{equation}
\parallel m(z;x,t|\sigma_d)^{-1}\parallel_{L^\infty(C\setminus(Z\cup\bar{Z}))}\lesssim 1. \label{normm}
\end{equation}

\end{proposition}
\begin{proof}
	The uniqueness of solution follows from Liouville's Theorem. And the symmetries in \textbf{problem 1} means that the solution of it must admit a partial fraction expansion of following from
	\begin{equation}
	m(z;x,t|\sigma_d)=\left(\begin{array}{cc}
	1 & 0\\
	\alpha & 0
	\end{array}\right)+\sum_{k=1}^{N}\left[ \frac{1}{z-z_k}\left(\begin{array}{cc}
	\nu_k(x,t) & 0\\
	\beta_k(x,t) & 0
	\end{array}\right)+\frac{1}{z-\bar{z}_k}\left(\begin{array}{cc}
	0 & -z_k^{-1}\overline{\beta_k(x,t)}\\
	0 & \overline{\nu_k(x,t)}
	\end{array}\right)\right] .\label{EXPM}
	\end{equation}
	Follow the prove in The. 4.3 of \cite{Liu2} we similar have the existence of solution.
	Since det($m(z;x,t|\sigma_d)$)=1, we only need to consider $\parallel m(z;x,t|\sigma_d)\parallel_{L^\infty(C\setminus(Z\cup\bar{Z}))}$. And from (\ref{EXPM})	we simply obtain the   consequence.
\end{proof}

From Trace formula we have
\begin{equation}
a(z)=\exp\left[-\frac{1}{2\pi i}\int_R\frac{\log[1+\zeta|r(\zeta)|^2]}{\zeta-z}d\zeta\right]\prod_{n=1}^{N}\frac{z-z_n}{z-\bar{z}_n},\hspace{0.5cm}z\in C^+.\label{trace}\\
\end{equation}
Let $\triangle \subseteqq\left\lbrace 1,2,...,N\right\rbrace $, and define
\begin{equation}
a_\triangle(z)=\prod_{k\in\triangle }\dfrac{z-z_k}{z-\bar{z}_k}\exp\left[-\frac{1}{2\pi i}\int_R\frac{\log[1+\zeta|r(\zeta)|^2]}{\zeta-z}d\zeta\right].
\end{equation}
The renormalization
\begin{equation}
m^\triangle(z|\sigma_d)=m(z|\sigma_d)a^\triangle(z)^{\sigma_3},
\end{equation}
splits the poles according to the choice of $\triangle$, and it satisfies the following modified discrete Riemann-Hilbert problem.

\textbf{problem 2}. Given discrete data $\sigma_d=\left\lbrace(z_k,c_k) \right\rbrace _{k=1}^N$, and $\triangle \subseteqq\left\lbrace 1,2,...,N\right\rbrace $.Find a matrix-valued function z$\in C \rightarrow m^\triangle(z;x,t|\sigma_d)$ with following properties:

$\bullet$ Analyticity: $m^\triangle(z;x,t|\sigma_d)$ is meromorphic in $\mathbb{C}\setminus (\Sigma^{(2)}\cup Z\cup \bar{Z})$;

$\bullet$ Symmetry: $m^\triangle_{22}(z;x,t|\sigma_d)=\overline{m^\triangle_{11}(\bar{z};x,t|\sigma_d)}$, $m^\triangle_{21}(z;x,t|\sigma_d)=-z\overline{m^\triangle_{12}(\bar{z};x,t|\sigma_d)}$, which means $m^\triangle(z;x,t|\sigma_d)=z^{\sigma_3/2}\sigma_2\overline{m^\triangle(\bar{z};x,t|\sigma_d)}\sigma_2^{-1}z^{-\sigma_3/2}$;

$\bullet$ Asymptotic behaviors:
\begin{align}
&m^\triangle(z;x,t|\sigma_d) \sim \left(\begin{array}{cc}
1 & 0\\
\alpha & 0
\end{array}\right)+O(z^{-1}),\hspace{0.5cm}z \rightarrow \infty
\end{align}
for a constant $\alpha$ determined by the symmetry condition above;

$\bullet$ Residue conditions: $m(z;x,t|\sigma_d)$ has simple poles at each point in $ Z\bigcup \bar{Z}$ satisfying
\begin{align}
&\res_{z=z_n}m^\triangle(z;x,t|\sigma_d)=\lim_{z\to z_n}m^\triangle(z;x,t|\sigma_d)\tau^\triangle_k,\\
&\res_{z=\bar{z}_n}m^\triangle(z;x,t|\sigma_d)=\lim_{z\to \bar{z}_n}m^\triangle(z;x,t|\sigma_d)\hat{\tau}_k^\triangle,
\end{align}
where $\tau_k$ is a nilpotent matrix satisfies
\begin{align}
&\tau_k^\triangle=\left\{\begin{array}{ll}
\left(\begin{array}{cc}
0 & 0\\
\gamma_ka^\triangle(z)^2 & 0
\end{array}\right) & k\notin \triangle\\
\left(\begin{array}{cc}
0 & \gamma_k^{-1}a'^\triangle(z)^{-2}\\
0 & 0
\end{array}\right)  & k\in \triangle\\
\end{array}\right.,\hspace{0.5cm}\hat{\tau}_k^\triangle=z_k^{\sigma_3/2}\sigma_2\overline{\tau}_k^\triangle\sigma_2^{-1}z_k^{-\sigma_3/2},\nonumber\\
&\gamma_k=c_ke^{2i[(z-\rho)x+2(z-\rho)^2t]}.\label{tau}
\end{align}
Since $m^\triangle(z;x,t|\sigma_d)$ is a explicit transformation of $m(z;x,t|\sigma_d)$, from \textbf{proposition} \ref{unim} we obtain the existence uniqueness of the solution of \textbf{problem 2}. If $q_{sol}(x,t)=q_{sol}(x,t;\sigma_d)$ denotes the $N$-soliton solution of (\ref{MNLS1}) encoded by \textbf{problem 1}, we also have the reconstruction formula that
\begin{equation}
q_{sol}(x,t)=\lim_{z\to \infty}2iz(m^\triangle(z;x,t|\sigma_d))_{12},
\end{equation}
which show that each normalization encodes $q_{sol}(x,t)$ in the same way. If we choosing $\triangle$ appropriately, the asymptotic limits $|t|\to\infty$ with $z_0=-x/4t+\rho$ bounded are under better asymptotic control. Then we consider the long-time behavior of soliton solutions.

Give pairs points $x_1\leq x_2\in R$ and velocities $v_1\leq v_2 \in R$, and define the cone
\begin{equation}
C(x_1,x_2,v_1,v_2)=\left\lbrace (x,t)\in R^2|x=x_0+vt\text{ ,with }x_0\in[x_1,x_2]\text{, }v\in[v_1+4\rho,v_2+4\rho]\right\rbrace.\label{coneC}
\end{equation}
Denote $I=[-v_2/4+\rho,-v_1/4+\rho]$, and let
\begin{align}
&Z(I)=\left\lbrace z_k\in Z|Rez_k\in I\right\rbrace ,\hspace{2.1cm}N(I)=|Z(I)|\nonumber,\\
&Z^-(I)=\left\lbrace z_k\in Z|Rez_k< -v_2/4+\rho\right\rbrace,\hspace{0.5cm}Z^+(I)=\left\lbrace z_k\in Z|Rez_k> -v_1/4+\rho\right\rbrace\nonumber,\\
&c_k(I)=c_k\prod_{Rez_n\in I_{z_0}^\eta\setminus I}\left( \frac{z_k-z_n}{z_k-\bar{z}_n}\right) ^2\exp\left[-\frac{1}{\pi i}\int_{I_{z_0}^\eta}\frac{\log[1+\zeta|r(\zeta)|^2]}{\zeta-z}d\zeta\right].\label{dataI}
\end{align}

\begin{figure}[H]
	\centering
	\centering
	\subfigure[]{
		\begin{tikzpicture}[node distance=2cm]
		\filldraw[yellow,line width=2] (-1,-2.4) rectangle (1,2.4);
		\draw[->](-2.5,0)--(2.5,0)node[right]{Re$z$};
		\draw[->](1,-2.5)--(1,2.5)node[above]{$-v_2/4$};
		\draw[->](-1,-2.5)--(-1,2.5)node[above]{$-v_1/4$};
		\draw[->](0,0)--(-0.8,0);
		\draw[->](-0.8,0)--(-1.8,0);
		\draw[->](0,0)--(0.8,0);
		\draw[->](0.8,0)--(1.8,0);
		\coordinate (A) at (0.5,1.2);
		\coordinate (B) at (0.6,-1.2);
		\coordinate (C) at (1.5,1.6);
		\coordinate (D) at (1.6,-1.6);
		\coordinate (E) at (-1.7,0.6);
		\coordinate (F) at (-1.7,-0.6);
		\coordinate (G) at (-0.3,1.9);
		\coordinate (H) at (-0.3,-1.9);
		\fill (A) circle (1pt) node[right] {$z_1$};
		\fill (B) circle (1pt) node[right] {$\bar{z}_1$};
		\fill (G) circle (1pt) node[left] {$z_2$};
		\fill (H) circle (1pt) node[left] {$\bar{z}_2$};
		\fill (C) circle (1pt) node[right] {$z_3$};
		\fill (D) circle (1pt) node[right] {$\bar{z}_3$};
		\fill (E) circle (1pt) node[right] {$z_4$};
		\fill (F) circle (1pt) node[right] {$\bar{z}_4$};
		\label{figC}
		\end{tikzpicture}
	}
	\subfigure[]{
		\begin{tikzpicture}[node distance=2cm]
		\draw[yellow, fill=yellow] (0.6,0)--(1.5,2.5)--(-1.0,2.5)--(-0.5,0)--(-1.4,-2.5)--(1.0,-2.5)--(0.6,0);
		\draw[->](-2.5,0)--(2.5,0)node[right]{x};
		\draw[->](0.6,0)--(1.5,2.5)node[right]{$x=v_2t+x_2$};
		\draw[->](-0.5,0)--(-1.4,-2.5)node[left]{$x=v_2t+x_1$};
		\draw[->](-0.5,0)--(-1.0,2.5)node[left]{$x=v_1t+x_1$};
		\draw[->](0.6,0)--(1.0,-2.5)node[right]{$x=v_1t+x_2$};
		\draw[->](0,0)--(-0.8,0);
		\draw[->](-0.8,0)--(-1.8,0);
		\draw[->](0,0)--(0.8,0);
		\draw[->](0.8,0)--(1.8,0);
		\coordinate (A) at (0.6,0);
		\coordinate (B) at (-0.5,0);
		\coordinate (I) at (0,0);
		\fill (A) circle (2pt) node[above] {$x_1$};
		\fill (B) circle (2pt) node[above] {$x_2$};
		\fill (I) circle (0pt) node[below] {$C$};
		\label{figzero}
		\end{tikzpicture}
	}
	\caption{(a) In the example here, the original data has four pairs zero points of discrete spectrum, but insider the cone C only three pairs points with $Z(I)=\left\lbrace z_1,z_2 \right\rbrace$; (b) The cone $C(x_1,x_2,v_1,v_2)$}
	\label{figC(I)}
\end{figure}
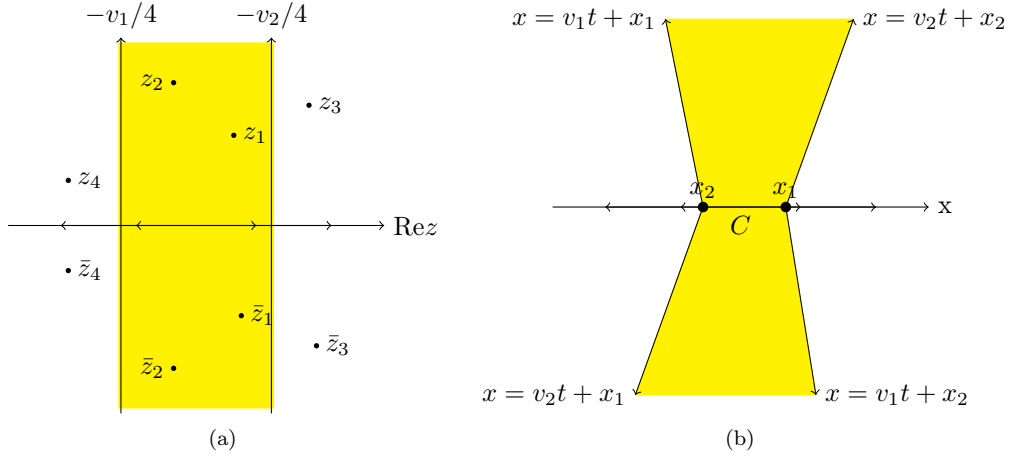

Then we have following lemma
\begin{lemma}
	Fix reflection-less data $D=\left\lbrace (z_k,c_k)\right\rbrace_{k=1}^N $, $D(I)=\left\lbrace (z_k,c_k(I))|z_k\in Z(I)\right\rbrace$. Then as $|t|\to\infty$ with $(x,t)\in C(x_1,x_2,v_1,v_2)$, we have
	\begin{equation}
	m(z;x,t|D)=\left( I+O(e^{-8\mu(I)|t|})\right) 	m(z;x,t|D(I)),
	\end{equation}
	where $\mu(I)=\min_{z_K\in Z\setminus Z(I)}\left\lbrace Im(z_k) dist(Re(z_k),I)\right\rbrace $.
\end{lemma}
\begin{proof}
	Let $\triangle^+(I)=\left\lbrace k|Re(z_k)<-v_2/4+\rho\right\rbrace $, $\triangle^-(I)=\left\lbrace k|Re(z_k)>-v_1/4+\rho\right\rbrace $. Then we consider the case that $\triangle$ in \textbf{problem 1} is $\triangle^\eta(I)$, where $\eta=sgn(t)$. For $z\in Z\setminus Z(I)$ and $(x,t)\in C(x_1,x_2,v_1,v_2)$, denote $x=x_0+(v_0+4\rho)t$, where $x_0\in [x_1,x_2]$ and $v_0\in[v_1,v_2]$. Note that the residue coefficients (\ref{tau}) have that
	\begin{equation}
	|\tau_k^{\triangle^\pm(I)}|=|c_ke^{-2x_0Rez}||e^{-8tImz(Rez-\rho-v_0/4)}|.
	\end{equation}
	So it have following asymptotic property
	\begin{equation}
	\parallel\tau_k^{\triangle^\pm(I)}\parallel=O(e^{-8\mu(I)|t|}),\hspace{0.5cm}  t\to \pm\infty.\label{asytau}
	\end{equation}
	$D_k$ is a small disks centred in each $z_k\in Z\setminus Z(I)$ with  radius smaller than $\mu$. Denote $\partial D_k$ is the boundary of $D_k$. Then we can  introduce a new transformation which can remove the poles $z_k\in Z\setminus Z(I)$ and these residues change to near-identity jumps.
	\begin{equation}
	\tilde{m}^{\triangle^\eta(I)}(z;x,t|D)=\left\{\begin{array}{ll}
	m^{\triangle^\eta(I)}(z;x,t|D)\left(I-\frac{\tau^{\triangle^\eta(I)}_k}{z-z_k} \right)  & z\in D_k\\
	m^{\triangle^\eta(I)}(z;x,t|D)\left(I-\frac{\hat{\tau}^{\triangle^\eta(I)}_k}{z-\bar{z}_k} \right)  & z\in \bar{D}_k\\
	m^{\triangle^\eta(I)}(z;x,t|D)  & elsewhere\\
	\end{array}\right.
	\end{equation}
	Comparing with $m^{\triangle^\eta(I)}(z|D)$, the new function has new jump in each $\partial D_k$ which denote by $\tilde{V}(Z)$. Then using (\ref{asytau})  we have
	\begin{equation}
		\parallel\tilde{V}(Z)-I\parallel_{L^\infty(\tilde{\Sigma})}=O(e^{-8\mu(I)|t|}),\hspace{0.5cm}\tilde{\Sigma}=\bigcup_{z_k\in Z\setminus Z(I)}\left( \partial D_k\cup\partial \bar{D}_k\right).\label{asytV}
	\end{equation}
	After transformation $\tilde{m}^{\triangle^\eta(I)}(z|D)$ has same poles and residue conditions as $m(z;x,t|D(I))$. So we denote $m_0(z)=\tilde{m}^{\triangle^\eta(I)}(z|D)m(z|D(I))^{-1}$, which has no poles. And it has jump matrix for $z\in\tilde{\Sigma}$,
	\begin{equation}
	m_0^+(z)=m_0^-(z)V_{m_0}(z),\hspace{0.5cm}V_{m_0}(z)=m(z|D(I))\tilde{V}(Z)m(z|D(I))^{-1}.
	\end{equation}
	From (\ref{asytV}) and (\ref{normm}) applied to $m(z|D(I))$ the theory of small norm \textbf{RHPs} \cite{RN4},\cite{RN5}, we have $m_0(z)$ exists and $m_0(z)=I+O(e^{-8\mu(I)|t|})$ for $t\to \pm\infty$. Then we have the consequence.
\end{proof}
Using reconstruction formula  to $m(z;x,t|D)$ we  immediately have following  corollary.

\begin{corollary}
	Let $q_{sol}(x,t;D)$ and $q_{sol}(x,t;D(I))$ denote the $N$-soliton solution of (\ref{MNLS1}) corresponding to discrete scattering data D and D(I) respectively. And I, $C(x_1,x_2,v_1,v_2)$, $D(I)$ is given above. As $|t|\to\infty$ with $(x,t)\in C(x_1,x_2,v_1,v_2)$, we have
	\begin{equation}
	\lim_{z\to \infty}2iz(m(z;x,t|\sigma_d))_{12}=q_{sol}(x,t;D)=q_{sol}(x,t;D(I))+O(e^{-4\mu(I)|t|}).
	\end{equation}
\end{corollary}
Then we back to the outer model and obtain following corollary.
\begin{corollary}
There exist a unique solution $M^{(out)}$ of \textbf{Riemann-Hilbert problem 6}	with
\begin{align}
M^{(out)}(z)&=m^{\triangle^-_{z_0,\eta}}(z|D^{(out)})\\
&=M(z;x,t|D(I))\prod_{Rez_n\in I_{z_0}^\eta\setminus I}\left( \frac{z-z_n}{z-\bar{z}_n}\right) ^{\sigma_3}+O(e^{-4\mu(I)|t|}),\label{mout}
\end{align}
where $D^{(out)}=\left\lbrace z_k,c_k(z_0)\right\rbrace_{k=1}^N $, $M(z;x,t|D(I))$ is the solution of\textbf{ Problem 1}, with $D(I)=\left\lbrace (z_k,c_k(I))|z_k\in Z(I)\right\rbrace$ and
\begin{equation*}
c_k(z_0)=c_k\exp\left[-\frac{1}{\pi i}\int_{I_{z_0}^\eta}\frac{\log[1+\zeta|r(\zeta)|^2]}{\zeta-z}d\zeta\right].
\end{equation*}
Then substitute (\ref{mout})  into (\ref{normm}) we immediately have
\begin{equation}
\parallel M^{(out)}(z)^{-1}\parallel_{L^\infty(C\setminus(Z\cup\bar{Z}))}\lesssim 1.\label{normmout}
\end{equation}
Moreover, we have reconstruction formula
\begin{equation}
\lim_{z\to \infty}2iz(M^{(out)})_{12}=q_{sol}(x,t;D^{(out)}),\label{expMout}
\end{equation}
where the $q_{sol}(x,t;D^{(out)})$ is the $N$-soliton solution of (\ref{MNLS1}) corresponding to discrete scattering data $D^{(out)}$. And
\begin{equation}
q_{sol}(x,t;D^{(out)})=q_{sol}(x,t;D(I))+O(e^{-4\mu(I)|t|}),\hspace{0.5cm} \text{for }t\to\pm\infty.\label{asyqout}
\end{equation}
\end{corollary}
As we can find that  $M(z;x,t|D(I))$ is the solution of  with  reflectionless  scattering data case  $\left\lbrace r\equiv0;\left\lbrace (z_k,c_k(I))|z_k\in Z(I)\right\rbrace\right\rbrace $,  combine (\ref{intq}), (\ref{varphi-rho}) and (\ref{arho}) we have following proposition.
\begin{proposition}\label{Msolrho}
	The unique solution $M(z;x,t|D(I))$ of \textbf{ problem 1} have that
	\begin{align}
	[M(\rho;x,t|D(I))_+]_{11}&=e^{i/2\int_{-\infty}^{x}|q_{sol}(x,t|D(I))|^2dy},\\
	[M(\rho;x,t|D(I))_+]_{12}&=-e^{i/2\int_{-\infty}^{x}|q_{sol}(x,t|D(I))|^2dy}\nonumber\\
	&\int_{-\infty}^{x}q_{sol}(y,t|D(I))e^{-i\int_{-\infty}^{y}|q_{sol}(s,t|D(I))|^2ds}dy,	
	\end{align}
	and  $[M_+(\rho;x,t|D(I))]_{22}=\overline{[M_+(\rho;x,t|D(I))]_{11}}$, $[M_+(\rho;x,t|D(I))]_{21}=-\rho\overline{[M_+(\rho;x,t|D(I))]_{12}}$.
\end{proposition}

\section{Local model Riemann-Hilbert problem near $z_0$}
\quad From \textbf{proposition} \ref{pro3v2} we find that $V^{(2)}-I$ doesn't have a uniformly small jump for large time. So in the neighborhood $U_{z_0}$ of $z_0$, we establish local model to arrive at a uniformly small jump Riemann-Hilbert problem for function $E(Z)$. Notice that $\forall p \in Z\cup\bar{Z}$, $p\notin U_{z_0}$. Let $\xi=\xi(z)$ denote the local variable
\begin{equation}
\xi(z)=\sqrt{|8t|}(z-z_0).
\end{equation}
Then we have
\begin{equation*}
2t\theta=\xi^2/2-4t(z_0-\rho)^2,\hspace{0.5cm}(\eta(z-z_0))^{2i\eta k(z_0)}=(\eta\xi)^{2i\eta k(z_0)}e^{-i\eta k(z_0)log|8t|}.
\end{equation*}
Let
\begin{equation}
r_{z_0}=-r(z_0)T_0(z_0)^2e^{-i\eta k(z_0)log|8t|}e^{4ti(z_0-\rho)^2},\hspace{0.5cm}s_{z_0}=z\bar{r}_{z_0},
\end{equation}
then we have $1+r_{z_0}s_{z_0}=1+z|r(z_0)|^2\neq 0$. And since $1-X_Z(z)=1$ for $z\in U_{z_0}$, the jump matrix $V^{(2)}$ limit in $U_{z_0}$ denoted as $V^{(pc)}$ has become
\begin{equation}
V^{(pc)}=\left\{\begin{array}{llll}
\left(\begin{array}{cc}
1 & 0\\
s_{z_0}(\eta\xi)^{-2i\eta k(z_0)}e^{\xi^2i/2} & 1
\end{array}\right) & z\in \Sigma_1\\
\left(\begin{array}{cc}
1 & \dfrac{r_{z_0}}{1+r_{z_0}s_{z_0}}(\eta\xi)^{2i\eta k(z_0)}e^{-\xi^2i/2}\\
0 & 1
\end{array}\right)  &z\in \Sigma_2\\
\left(\begin{array}{cc}
1 & 0\\
\dfrac{s_{z_0}}{1+r_{z_0}s_{z_0}}(\eta\xi)^{-2i\eta k(z_0)}e^{\xi^2i/2} & 1
\end{array}\right)  &z\in \Sigma_3\\
\left(\begin{array}{cc}
1 & r_{z_0}(\eta\xi)^{2i\eta k(z_0)}e^{-\xi^2i/2}\\
0 & 1
\end{array}\right)  &z\in \Sigma_4\\
\end{array}\right..
\end{equation}
Then we have following \textbf{RHP 6 } which doesn't possess the symmetry condition shared by \textbf{RHP 4}, because we only use this model for  bounded values z.

\textbf{Parabolic Cylinder Model Riemann-Hilbert problem 6} Find an analytic function $M^{(pc)}(z;z_0,\eta)$: $C\setminus\Sigma^{(2)}\to SL_2(C)$ such that

$\bullet$ Analyticity: $M^{pc}(z;z_0,\eta)$ is meromorphic in $\mathbb{C}\setminus (\Sigma^{(2)}\cup Z\cup \bar{Z})$;

$\bullet$ Jump condition: $M^{pc}$ has continuous boundary values $M^{pc}_\pm$ on $R$ and
\begin{equation}
M^{pc}_+(z;z_0,\eta)=M^{pc}_-(z;z_0,\eta)V^{(pc)}(z;z_0,\eta),\hspace{0.5cm}z \in \Sigma^{(2)};\label{jumppc}
\end{equation}

$\bullet$ Asymptotic behaviors:
\begin{align}
&M^{pc}(z;z_0,\eta) \sim I+O(z^{-1}),\hspace{0.5cm}z \rightarrow \infty.
\end{align}
The precise details of the construction for this solution, which differ only slightly from the construction for KdV\cite{RN6}. In fact,This type of model problem is typical in integrable systems whenever there is a phase function. Here in our system is $\theta$, which has a quadratic critical point along the real line.  Here we only give the necessary details, and to arrive our consequence, we also need its boundedness property shown in following \textbf{Lemma} which proof can be find in \cite{RN8},Appendix D.

\begin{proposition}\label{Mpc}
The solution of \textbf{RHP 6 }is given as follow:
\begin{align}
M^{pc}(z;z_0,+)=F(\xi(z);r_{z_0},s_(z_0)),\nonumber\\
M^{pc}(z;z_0,-)=\sigma_2F(-\xi(z);r_{z_0},s_(z_0))\sigma_2,\label{mpc}
\end{align}
where $\kappa=k(z_0)$
\begin{equation}
F(\xi;r,s)=\Phi(\xi;r,s)P(\xi;r,s)e^{i\xi^2\sigma_3/4}\xi^{-i\kappa\sigma_3}
\end{equation}
\begin{equation}
P(\xi;r,s)=\left\{\begin{array}{llll}
\left(\begin{array}{cc}
1 & 0\\
s & 1
\end{array}\right) & \xi\in \Omega_1\\
\left(\begin{array}{cc}
1 & \dfrac{r}{1+rs}\\
0 & 1
\end{array}\right)  & \xi\in \Omega_3\\
\left(\begin{array}{cc}
1 & 0\\
\dfrac{-s}{1+rs} & 1
\end{array}\right)  & \xi\in \Omega_4\\
\left(\begin{array}{cc}
1 & -r\\
0 & 1
\end{array}\right)  & \xi\in \Omega_6\\
I  &z\in \Omega_2\cup\Omega_5\\
\end{array}\right.,
\end{equation}
\begin{equation}
\Phi(\xi;r,s)=\left\{\begin{array}{llll}
\left(\begin{array}{cc}
e^{-3\pi\kappa/4}D_{i\kappa}(\xi e^{-3i\pi/4}) & -i\beta_{12}e^{\pi(\kappa-i)/4}D{-i\kappa-1}(\xi e^{-i\pi/4})\\
i\beta_{21}e^{-3\pi(\kappa+i)/4}D{i\kappa-1}(\xi e^{-3i\pi/4}) & e^{\pi\kappa/4}D_{-i\kappa}(\xi e^{-i\pi/4})
\end{array}\right) & \xi\in C^+\\
\left(\begin{array}{cc}
e^{\pi\kappa/4}D_{i\kappa}(\xi e^{i\pi/4}) & -i\beta_{12}e^{-3\pi(\kappa-i)/4}D{-i\kappa-1}(\xi e^{3i\pi/4})\\
i\beta_{21}e^{\pi(\kappa+i)/4}D{i\kappa-1}(\xi e^{i\pi/4}) & e^{-3\pi\kappa/4}D_{-i\kappa}(\xi e^{3i\pi/4})
\end{array}\right)  & \xi\in C^-\\
\end{array}\right..
\end{equation}	
Here $D_a(a)$ denote the parabolic cylinder functions, and $\beta_{21}$ and $\beta_{12}$ are complex constants
\begin{equation}
\beta_{12}=\beta_{12}(r,s)=\dfrac{\sqrt{2\pi}e^{-\kappa\pi/2}e^{i\pi/4}}{s\Gamma(-i\kappa)},\hspace{0.5cm}\beta_{21}=\beta_{21}(r,s)=\dfrac{-\sqrt{2\pi}e^{-\kappa\pi/2}e^{-i\pi/4}}{r\Gamma(i\kappa)}.
\end{equation}
Then we consider the asymptotic behavior of the solution. When $\xi\to\infty$,
\begin{equation}
F(\xi;r,s)=I+\dfrac{1}{\xi}\left(\begin{array}{cc}
0 & -i\beta_{12}(r,s)\\
i\beta_{21}(r,s) & 0
\end{array}\right)+O(\xi^{-2}),
\end{equation}
from which we have
\begin{equation}
M^{pc}(z;z_0,\eta)=I+\frac{|8t|^{-1/2}}{z-z_0}\left(\begin{array}{cc}
0 & -iA_{12}(z_0,\eta)\\
iA_{21}(z_0,\eta) & 0
\end{array}\right)+O(|t|^{-1}),\hspace{0.5cm}z\in \partial U_{z_0},\label{asympc}
\end{equation}
with
\begin{align}
A_{12}(z_0,+)=\beta_{12}(r_{z_0},s_{z_0}),\hspace{0.5cm}A_{21}(z_0,+)=\beta_{21}(r_{z_0},s_{z_0}),\nonumber\\
A_{12}(z_0,-)=-\beta_{21}(r_{z_0},s_{z_0}),\hspace{0.5cm}A_{12}(z_0,-)=-\beta_{12}(r_{z_0},s_{z_0}),
\end{align}
satisfying following conditions
\begin{align*}
&|A_{12}(z_0,\eta)|^2=\dfrac{k(z_0)}{z_0},\hspace{0.5cm}A_{12}(z_0,\eta)=z_0\overline{A_{12}(z_0,\eta)},\\
arg\left(A_{12}(z_0,+) \right) =&\frac{\pi}{4}+arg\Gamma(ik(z_0))-arg(-z_0\overline{r(z_0)})+\frac{1}{\pi}\int_{-\infty}^{z_0}log|s-z_0|dslog(1+z|r(z)|^2)\\
&-k(z_0)log|8t|+4tz_0^2,\\
arg\left(A_{12}(z_0,-) \right) =&\frac{\pi}{4}-arg\Gamma(ik(z_0))-arg(-z_0\overline{r(z_0)})+\frac{1}{\pi}\int_{z_0}^{\infty}log|s-z_0|dslog(1+z|r(z)|^2)\\
&+k(z_0)log|8t|+4tz_0^2.
\end{align*}
And the $\infty$-norm of $M^{pc}$ has
\begin{equation}
\parallel M^{pc}(z;z_0,\eta)\parallel_{\infty}\lesssim 1,\hspace{0.5cm}\parallel M^{pc}(z;z_0,\eta)^{-1}\parallel_{\infty}\lesssim 1.
\end{equation}
\end{proposition}

Using (\ref{mpc}) we can define the local model
\begin{equation}
M^{(z_0)}(z)=M^{(out)}(z)M^{pc}(z;z_0,\eta),\hspace{0.5cm}z\in U_{z_0},
\end{equation}
which have the same jump matrix $V^{(2)}$ as $M^{RHP}$ from $M^{pc}$ and same residue conditions as $M^{RHP}$ from $M^{(out)}$. So we have that it is a bounded function in $U_{z_0}$.

\section{The small norm Riemann-Hilbert problem for E(z)}
\quad In this subsection we consider the error $E(z)$. From the definition of it (\ref{transm4}) and the analyticity of $M^{(out)}$ and $M^{(z_0)}$ we can obtain that $E(z)$ is analytic in $C\setminus \Sigma^{(E)}$, where
\begin{equation*}
 \Sigma^{(E)}=\partial U_{z_0}\cup (\Sigma^{(2)}\setminus U_{z_0}).
\end{equation*}
We will proof that for large times, the error $E(z)$ solves following small norm Riemann-Hilbert problem which we can expand asymptotically.

\textbf{Riemann-Hilbert problem 7}. Find a matrix-valued function $E(z)$: z$\in C \rightarrow m^\triangle(z;x,t|\sigma_d)$ with following properties:

$\bullet$ Analyticity: $E(z)$ is meromorphic in $\mathbb{C}\setminus (\Sigma^{(E)}\cup Z\cup \bar{Z})$;

$\bullet$ Symmetry: $E_{22}(z)=\overline{E_{11}(\bar{z})}$, $E_{21}(z)=-z\overline{E_{12}(\bar{z})}$, which means $E(z)=z^{\sigma_3/2}\sigma_2\overline{E(\bar{z})}\sigma_2^{-1}z^{-\sigma_3/2}$;

$\bullet$ Asymptotic behaviors:
\begin{align}
&E(z) \sim \left(\begin{array}{cc}
1 & 0\\
\alpha_E & 0
\end{array}\right)+O(z^{-1}),\hspace{0.5cm}|z| \rightarrow \infty
\end{align}
for a constant $\alpha_E$ determined by the symmetry condition above;

$\bullet$ Jump condition: $E$ has continuous boundary values $E_\pm$ on $\Sigma^{(E)}$ satisfy $E_+(z)=E_-(z)V^{(E)}$, where
\begin{equation}
V^{(E)}=\left\{\begin{array}{llll}
M^{(out)}(z)V^{(2)}(z)M^{(out)}(z)^{-1} & \xi\in \Sigma^{(2)}\setminus U_{z_0}\\
M^{(out)}(z)M^{pc}(z)^{-1}M^{(out)}(z)^{-1}  & \xi\in \partial U_{z_0}\\
\end{array}\right..\label{deVE}
\end{equation}

From \textbf{Proposition \ref{pro3v2}}, (\ref{asympc}), (\ref{normmout}), the jump matrix have
\begin{equation}
|V^{(E)}-I|=\left\{\begin{array}{llll}
e^{-4|t(z-z_0)^2|} & \xi\in \Sigma^{(2)}\setminus U_{z_0}\\
|t|^{-1/2}  & \xi\in \partial U_{z_0}\\
\end{array}\right.,\label{VE-I}
\end{equation}
which implies that for $k\in N$, $p\leq 1$
\begin{equation}
\parallel V^{(E)}-I\parallel_{L^{p,k}(R)\cup L^\infty(R)}=O(|t|^{-1/2}).\label{VE}
\end{equation}

This uniformly approach to zero as $|t| \to \infty$ of $V^{(E)}-I$ make \textbf{Riemann-Hilbert problem 7} are small-norm Riemann-Hilbert problem, for which there is a well known existence and uniqueness theorem\cite{RN5,RN9,RN10}. Then we have following Lemma.
\begin{lemma}
	There exists a unique solution $E(z)$ of \textbf{RHP 7} stratifies
	\begin{equation}
	\parallel E-\left(\begin{array}{cc}
	1 & 0\\
	\alpha_E & 0
	\end{array}\right)\parallel_{L^\infty(C\setminus\Sigma^{(E)})}\lesssim |t|^{-1/2}.
	\end{equation}
	And when $z\to \infty$,
	\begin{equation}
	E(z) \sim \left(\begin{array}{cc}
	1 & 0\\
	\alpha_E & 0
	\end{array}\right)+\frac{E_1}{z}+O(z^{-2}),\label{asyzE}
	\end{equation}
	$\bar{\alpha}_E=(E_1)_{12}$ and it has
	\begin{equation}
	2i(E_1)_{12}=|t|^{-1/2}f(x,t) +O(|t|^{-1}),\label{expE}
	\end{equation}
	where
	\begin{equation}
	f(x,t)=2^{-1/2}\left[ A_{12}(z_0,\eta)M_{11}^{(out)}(z_0)^2+A_{21}(z_0,\eta)M_{12}^{(out)}(z_0)^2\right].\label{f(x,t)}
	\end{equation}
\end{lemma}
	\begin{proof}
		Because \textbf{RHP 7} isn't the standard Riemann-Hilbert problem, we need to construct the solution $E(z)$ row-by-row. For the first row, we denote $e^1=(E_{11},E_{12})$, which have following property
		
		1. $e^1\sim$(1, 0) as $z\to\infty$;
		
		2.$e^1$ has continuous boundary values $e^1_\pm$ on $\Sigma^{(E)}$ satisfy $e^1_+(z)=e^1_-(z)V^{(E)}$.
		
		Then by Plemelj formula we have
		\begin{equation}
		e^1(z)=(1, 0)+\frac{1}{2\pi i}\int_{\Sigma^{(E)}}\dfrac{\left( (1,0)+\mu_1(s)\right) (V^{(E)}-I)}{s-z}ds,\label{e^1}
		\end{equation}
		where the $\mu_1\in L^2(\Sigma^{(E)})$ is the unique solution of following equation:
		\begin{equation}
		(1-C_E)\mu_1=C_E\left((1,0) \right).
		\end{equation}
		 $C_E$ is a integral operator defined as
		 \begin{equation}
		 C_E(f)(z)=C_-\left( f(V^{(E)}-I)\right) ,
		 \end{equation}
		 where the $C_-$ is the usual Cauchy projection operator on $\Sigma^{(E)}$
		 \begin{equation}
		 C_-(f)(s)=\lim_{z\to \Sigma^{(E)}_-}\frac{1}{2\pi i}\int_{\Sigma^{(E)}}\dfrac{f(s)}{s-z}ds.
		 \end{equation}
		 Then by (\ref{VE}) we have
		 \begin{equation}
		 \parallel C_E\parallel\leq\parallel C_-\parallel \parallel V^{(E)}-I\parallel_\infty \lesssim O(t^{-1/2}),
		 \end{equation}
		 which means for sufficiently large t, $\parallel C_E\parallel<1$. Then we have 1-$C_E$ is invertible, so the $\mu_1$ is unique existence. Moreover,
		 \begin{equation}
		 \parallel \mu_1\parallel_{L^2(\Sigma^{(E)})}\lesssim\dfrac{\parallel C_E\parallel}{1-\parallel C_E\parallel}\lesssim|t|^{-1/2}.\label{mu1}
		 \end{equation}
		 Then we can proof that $e^1$ nearly identity.
		 \begin{equation}
		 |e^1(z)-(1,0)|\leq|\frac{1}{2\pi i}\int_{\Sigma^{(E)}}\dfrac{ (1,0) (V^{(E)}-I)}{s-z}ds|+|\frac{1}{2\pi i}\int_{\Sigma^{(E)}}\dfrac{\mu_1(s) (V^{(E)}-I)}{s-z}ds|.
		 \end{equation}
		 When there exist a constant $d>0$ such that $\inf_{z\in\Sigma^{(E)}}|s-z|>d$, we have
		 \begin{equation}
		 |e^1(z)-(1,0)|\leq\dfrac{1}{2\pi d}\left(\parallel V^{(E)}-I\parallel _{L^1}+\parallel \mu_1\parallel _{L^2}\parallel V^{(E)}-I\parallel _{L^2} \right)\lesssim|t|^{-1/2}.
		 \end{equation}
		 And for z approaching to $\Sigma^{(E)}$, because the jump matrix on the contours $\Sigma^{(E)}$ are locally analytic, we can make a invertible transformation$e^1\to \tilde{e}^1$( for example,  inversion transformation of a  circle center at $z_0$ with radius of $\mu/6$), which transform $\Sigma^{(E)}$ to a new contours $\tilde{\Sigma}^{(E)}$ with different points of self-intersection. Similarly we have $|\tilde{e}^1(z)-(1,0)|$ is bounded on $\Sigma^{(E)}$, then we obtain the boundedness of $e^1$.
		
		 Then we consider the second row $e^2=(E_{21},E_{22})$, similarly it have following property
		
		 1. $e^2\sim$(1, $\alpha_E$) as $z\to\infty$;
		
		 2.$e^2$ has continuous boundary values $e^1_\pm$ on $\Sigma^{(E)}$ satisfy $e^2_+(z)=e^2_-(z)V^{(E)}$,
		
		 where from the symmetry $E_{21}(z)=-z\overline{E_{12}(\bar{z})}$ and (\ref{e^1}) making $z\to\infty$ we obtain
		 \begin{equation}
		 \alpha_E=-\left[ \frac{1}{2\pi i}\int_{\Sigma^{(E)}}\left( (1,0)+\mu_1(s)\right) (V^{(E)}-I)ds\right] _2,
		 \end{equation}
		 where the subscript 2 means the second element. By (\ref{VE}) and (\ref{mu1}) we have
		 \begin{equation}
		 |\alpha_E|\lesssim|t^{-1/2}|\label{aE}.
		 \end{equation}
		 In the same way we obtain
		 \begin{equation}
		 e^2(z)=(\alpha_E, 1)+\frac{1}{2\pi i}\int_{\Sigma^{(E)}}\dfrac{\left( (\alpha_E, 1)+\mu_2(s)\right) (V^{(E)}-I)}{s-z}ds,\label{e^2}
		 \end{equation}
		 where the $\mu_2\in L^2(\Sigma^{(E)})$ is the unique solution of following equation:
		 \begin{equation}
		 (1-C_E)\mu_1=C_E\left((\alpha_E, 1) \right).
		 \end{equation}
		 Also we have $\parallel \mu_2\parallel_{L^2(\Sigma^{(E)})}\lesssim|t|^{-1/2}$ and $|e^2(z)-(\alpha_E, 1)|\lesssim|t|^{-1/2}$.
		 Now we denote $E=(e^1,e^2)^T$, $\mu=(\mu^1,\mu^2)^T$, then $E$ has expansion for $z\to\infty$
		 \begin{equation}
		 E=\left(\begin{array}{cc}
		 1 & 0\\
		 \alpha_E & 0
		 \end{array}\right)+\frac{E_1}{z}+O(z^{-2}),
		 \end{equation}
		 where
		 \begin{equation}
		 E_1=\frac{-1}{2\pi i}\int_{\Sigma^{(E)}}\left(\left(\begin{array}{cc}
		 1 & 0\\
		 \alpha_E & 0
		 \end{array}\right)+\mu(s)\right) (V^{(E)}-I)ds.
		 \end{equation}
		 Using (\ref{aE}), (\ref{VE}), (\ref{deVE}) (\ref{asympc}) and (\ref{mu1}) we have
		 \begin{align}
		 E_1=&\frac{-1}{2\pi i}\int_{\partial U_{z_0}} (V^{(E)}-I)ds+O(|t^{-1}|)\nonumber\\
		 =&M^{(out)}(Z_0)A(z_0,\eta)M^{(out)}(Z_0)^{-1}|8t|^{-1/2}+O(|t^{-1}|).
		 \end{align}
	\end{proof}
When estimating the solution $u(x,t)$ of the MNLS equation (\ref{MNLS}), we need the following result which provides the large-time behavior of the error term $E(\rho)$.
\begin{proposition}\label{Erho}
	When $|t|\to \infty$, the unique solution $E(z)$ of \textbf{RHP 7} described by  above Lemma  satisfies:
	
	\textbf{1.} when $\rho\in U_{z_0}$,
	\begin{align}
	&E_{11}(\rho)=1-\nonumber\\
	&\sum_{s=z_0,\rho}\frac{|t|^{-1/2}}{4\sqrt{2}\pi }\left[ A_{21}(z_0,\eta)M^{(out)}_{12}(s)\overline{M^{(out)}_{11}(s)}-sA_{12}(z_0,\eta)M^{(out)}_{11}(s)\overline{M^{(out)}_{12}(s)}\right] \nonumber\\
	&+O(|t|^{-1}),\\
	&E_{12}(\rho)=\sum_{s=z_0,\rho}\frac{|t|^{-1/2}}{4\sqrt{2}\pi }\left[ A_{21}(z_0,\eta)\overline{M^{(out)}_{11}(s)}^2+s^2A_{12}(z_0,\eta)\overline{M^{(out)}_{12}(s)}^2\right] \nonumber\\
	&+O(|t|^{-1});
	\end{align}
	\textbf{2.} when $\rho\notin U_{z_0}$,
	\begin{align}
	&E_{11}(\rho)=1-\nonumber\\
	&\frac{|t|^{-1/2}}{4\sqrt{2}\pi }\left[ A_{21}(z_0,\eta)M^{(out)}_{12}(z_0)\overline{M^{(out)}_{11}(z_0)}-z_0A_{12}(z_0,\eta)M^{(out)}_{11}(z_0)\overline{M^{(out)}_{12}(z_0)}\right] \nonumber\\
	&+O(|t|^{-1}),\\
	&E_{12}(\rho)=\frac{|t|^{-1/2}}{4\sqrt{2}\pi }\left[ A_{21}(z_0,\eta)\overline{M^{(out)}_{11}(z_0)}^2+z_0^2A_{12}(z_0,\eta)\overline{M^{(out)}_{12}(z_0)}^2\right] \nonumber\\
	&+O(|t|^{-1}).
	\end{align}
	$A_{12}(z_0,\eta)$ and $A_{12}(z_0,\eta)$ is given in\textbf{ Proposition \ref{Mpc}}. Together with $E_{22}(\rho)=\overline{E_{11}(\rho)}$, $E_{21}(\rho)=-z\overline{E_{12}(\rho)}$ we obtain the long-time behavior of $E(\rho)$.
\end{proposition}
\begin{proof}
	We only calculate $e_1(\rho)$, because $e_2(\rho)$ can obtain by symmetry of $E(z)$. From (\ref{e^1}) and (\ref{deVE}) we have
	\begin{align}
	e_1(\rho)-(1, 0)&=\frac{1}{2\pi i}\int_{\Sigma^{(E)}}\dfrac{\left( (1,0)+\mu_1(s)\right) (V^{(E)}-I)}{s-\rho}ds\nonumber\\
	&=\frac{1}{2\pi i}\int_{\Sigma^{(2)}\setminus U_{z_0}}\dfrac{\left( (1,0)+\mu_1(s)\right) }{s-\rho}M^{(out)}(V^{(2)}-I)[M^{(out)}]^{-1}ds\nonumber\\
	&+\frac{1}{2\pi i}\int_{\partial U_{z_0}}\dfrac{\left( (1,0)+\mu_1(s)\right) }{s-\rho}M^{(out)}(M^{pc}-I)[M^{(out)}]^{-1}ds.
	\end{align}
	For the first integral, we calculate on $\Sigma_1\setminus U_{z_0}$. The others are similar. Let $s=z_0+le^{i\pi/4}$, $l\in[\mu/3,+\infty]$, and note that $|s-\rho|>\sqrt{2}\mu/6$. Then together with (\ref{VE-I}) and (\ref{mu1}) we obtain
	\begin{align}
	&\parallel \int_{\Sigma^{(2)}\setminus U_{z_0}}\dfrac{\left( (1,0)+\mu_1(s)\right) (V^{(E)}-I)}{s-\rho}ds\parallel_\infty\nonumber\\
	&\lesssim\int_{\mu/3}^{+\infty}e^{-4|tl^2|}dl+\int_{\mu/3}^{+\infty}|\mu_1(l)e^{-4|tl^2|}dl\nonumber\\
	&\lesssim|t|^{-1}+|t|^{-1/2}\parallel \mu_1\parallel_{L^2}\lesssim|t|^{-1}.
	\end{align}
	For the second  integral, from (\ref{asympc}) and (\ref{mu1}) we have
	\begin{align}
	&\frac{1}{2\pi i}\int_{\partial U_{z_0}}\dfrac{\left( (1,0)+\mu_1(s)\right) }{s-\rho}M^{(out)}(M^{pc}-I)[M^{(out)}]^{-1}ds\nonumber\\
	&=-\frac{|t|^{-1/2}}{4\sqrt{2}\pi i}\int_{\partial U_{z_0}}\dfrac{ (1,0) }{(s-\rho)(s-z_0)}M^{(out)}A(z_0,\eta)[M^{(out)}]^{-1}ds+O(|t|^{-1}),
	\end{align}
	where
	\begin{align}
	&(1,0)M^{(out)}A(z_0,\eta)[M^{(out)}]^{-1}\nonumber\\
	&=i\left( A_{21}M^{(out)}_{12}M^{(out)}_{22}+A_{12}M^{(out)}_{11}M^{(out)}_{21}, A_{21}[M^{(out)}_{22}]^2+A_{12}[M^{(out)}_{21}]^2\right) .
	\end{align}
	Then by residue theorem and the symmetry we obtain the result.
\end{proof}

Now combine above Lemmas and proposition about the boundedness of $M^{(out)}$, $E$, $M^{pc}$ we have
\begin{proposition}
	$M^{RHP}$ is the unique solution of \textbf{RHP 4}, which is in $L^\infty\left( C\setminus(\Sigma^{(2)}\cap supp(1-X_Z))\right) $ and
	\begin{equation}
	\parallel M^{RHP}(z)^{\pm1}\parallel_{L^\infty(\Sigma^{(2)}\cap supp(1-X_Z)}\lesssim 1.\label{normMRHp}
	\end{equation}
\end{proposition}
And to estimate the gauge factor of the solution $u(x,t)$ of the MNLS equation (\ref{MNLS}), we also need the large-time behavior of  $M^{RHP}_+(\rho)$.
\begin{proposition}\label{MRHPrho}
	$M^{RHP}$ is the unique solution of \textbf{RHP 4}, then as $|t|\to\infty$,  $M^{RHP}_+(\rho)$ have
	
	\textbf{1.} when $\rho\in U_{z_0}$,
	\begin{align}
	[M^{RHP}_+(\rho)]_{11}&=[M_+(\rho;x,t|D(I))]_{11}\prod_{Rez_n\in I_{z_0}^\eta\setminus I}\left( \frac{\rho-z_n}{\rho-\bar{z}_n}\right)\nonumber\\
	&+|t|^{-1/2}G_1(\rho,z_0;x,t|D(I)+O(|t|^{-1});\\
	[M^{RHP}_+(\rho)]_{12}&=[M_+(\rho;x,t|D(I))]_{12}\prod_{Rez_n\in I_{z_0}^\eta\setminus I}\left( \frac{\rho-\bar{z}_n}{\rho-z_n}\right)\nonumber\\
	&+|t|^{-1/2}G_2(\rho,z_0;x,t|D(I))+O(|t|^{-1}),
	\end{align}
	where
	\begin{align}
	&G_1(\rho,z_0;x,t|D(I))=\nonumber\\
	&-[M_+(\rho;x,t|D(I))]_{11}\prod_{Rez_n\in I_{z_0}^\eta\setminus I}\left( \frac{\rho-z_n}{\rho-\bar{z}_n}\right)\sum_{s=z_0,\rho}\frac{1}{4\sqrt{2}\pi }\nonumber\\
	&A_{21}(z_0,\eta)[M_+(\rho;x,t|D(I))]_{12}\overline{M_{11}(s;x,t|D(I))}\prod_{Rez_n\in I_{z_0}^\eta\setminus I}\left( \frac{s-\bar{z}_n}{s-z_n}\right)^2\nonumber\\
	&-[M_+(\rho;x,t|D(I))]_{11}\prod_{Rez_n\in I_{z_0}^\eta\setminus I}\left( \frac{\rho-z_n}{\rho-\bar{z}_n}\right)\sum_{s=z_0,\rho}\frac{1}{4\sqrt{2}\pi }\nonumber\\
	&sA_{12}(z_0,\eta)M_{11}(s;x,t|D(I))\overline{M_{12}(s;x,t|D(I))}\prod_{Rez_n\in I_{z_0}^\eta\setminus I}\left( \frac{s-z_n}{s-\bar{z}_n}\right)^2\nonumber\\
	&-\rho\overline{[M_+(\rho;x,t|D(I))]_{12}}\prod_{Rez_n\in I_{z_0}^\eta\setminus I}\left( \frac{\rho-z_n}{\rho-\bar{z}_n}\right)\nonumber\\
	&\sum_{s=z_0,\rho}\frac{1}{4\sqrt{2}\pi } A_{21}(z_0,\eta)\overline{M_{11}(s;x,t|D(I))}^2\prod_{Rez_n\in I_{z_0}^\eta\setminus I}\left( \frac{s-\bar{z}_n}{s-z_n}\right)^2\nonumber\\
	&-\rho\overline{[M_+(\rho;x,t|D(I))]_{12}}\prod_{Rez_n\in I_{z_0}^\eta\setminus I}\left( \frac{\rho-z_n}{\rho-\bar{z}_n}\right)\nonumber\\
	&\sum_{s=z_0,\rho}\frac{1}{4\sqrt{2}\pi }s^2A_{12}(z_0,\eta)\overline{M_{12}(s;x,t|D(I))}^2\prod_{Rez_n\in I_{z_0}^\eta\setminus I}\left( \frac{s-z_n}{s-\bar{z}_n}\right)^2;\label{G1}
	\end{align}
	and
	\begin{align}
	&G_2(\rho,z_0;x,t|D(I))=-[M_+(\rho;x,t|D(I))]_{12}\prod_{Rez_n\in I_{z_0}^\eta\setminus I}\left( \frac{\rho-\bar{z}_n}{\rho-z_n}\right)\nonumber\\
	&\sum_{s=z_0,\rho}\frac{1}{4\sqrt{2}\pi } A_{21}(z_0,\eta)M_{12}(s;x,t|D(I))\overline{M_{11}(s;x,t|D(I))}\prod_{Rez_n\in I_{z_0}^\eta\setminus I}\left( \frac{s-\bar{z}_n}{s-z_n}\right)^2\nonumber\\
	&+[M_+(\rho;x,t|D(I))]_{12}\prod_{Rez_n\in I_{z_0}^\eta\setminus I}\left( \frac{\rho-\bar{z}_n}{\rho-z_n}\right)\sum_{s=z_0,\rho}\frac{1}{4\sqrt{2}\pi }\nonumber\\
	&sA_{12}(z_0,\eta)M_{11}(s;x,t|D(I))\overline{M_{12}(s;x,t|D(I))}\prod_{Rez_n\in I_{z_0}^\eta\setminus I}\left( \frac{s-z_n}{s-\bar{z}_n}\right)^2\nonumber\\
&+\overline{[M_+(\rho;x,t|D(I))]_{11}}\prod_{Rez_n\in I_{z_0}^\eta\setminus I}\left( \frac{\rho-z_n}{\rho-\bar{z}_n}\right)\nonumber\\
&\sum_{s=z_0,\rho}\frac{1}{4\sqrt{2}\pi } A_{21}(z_0,\eta)\overline{M_{11}(s;x,t|D(I))}^2\prod_{Rez_n\in I_{z_0}^\eta\setminus I}\left( \frac{s-\bar{z}_n}{s-z_n}\right)^2\nonumber\\
&+\overline{[M_+(\rho;x,t|D(I))]_{11}}\prod_{Rez_n\in I_{z_0}^\eta\setminus I}\left( \frac{\rho-z_n}{\rho-\bar{z}_n}\right)\nonumber\\
&\sum_{s=z_0,\rho}\frac{1}{4\sqrt{2}\pi }s^2A_{12}(z_0,\eta)\overline{M_{12}(s;x,t|D(I))}^2\prod_{Rez_n\in I_{z_0}^\eta\setminus I}\left( \frac{s-z_n}{s-\bar{z}_n}\right)^2\label{G2}.
	\end{align}

	\textbf{2.} when $\rho\notin U_{z_0}$,
	\begin{align}
	[M^{RHP}_+(\rho)]_{11}&=[M_+(\rho;x,t|D(I))]_{11}\prod_{Rez_n\in I_{z_0}^\eta\setminus I}\left( \frac{\rho-z_n}{\rho-\bar{z}_n}\right)\nonumber\\
	&+|t|^{-1/2}H_1(\rho,z_0;x,t|D(I)+O(|t|^{-1});\\
	[M^{RHP}_+(\rho)]_{12}&=[M_+(\rho;x,t|D(I))]_{12}\prod_{Rez_n\in I_{z_0}^\eta\setminus I}\left( \frac{\rho-\bar{z}_n}{\rho-z_n}\right)\nonumber\\
	&+|t|^{-1/2}H_2(\rho,z_0;x,t|D(I))+O(|t|^{-1}),
	\end{align}
	where
	\begin{align}
	&H_1(\rho,z_0;x,t|D(I))=\nonumber\\
	&-|[M_+(\rho;x,t|D(I))]_{11}|^2\prod_{Rez_n\in I_{z_0}^\eta\setminus I}\left( \frac{\rho-\bar{z}_n}{\rho-z_n}\right)\frac{1}{4\sqrt{2}\pi } A_{21}(z_0,\eta)[M_+(\rho;x,t|D(I))]_{12}\nonumber\\
	&-[M_+(\rho;x,t|D(I))]_{11}^2\prod_{Rez_n\in I_{z_0}^\eta\setminus I}\left( \frac{\rho-z_n}{\rho-\bar{z}_n}\right)^3\frac{1}{4\sqrt{2}\pi }\rho A_{12}(z_0,\eta)\overline{M_{12}(\rho ;x,t|D(I))}\nonumber\\
	&-\frac{1}{4\sqrt{2}\pi } A_{21}(z_0,\eta)\overline{M_{11}(\rho ;x,t|D(I))}^2\rho\overline{M_{12}(\rho ;x,t|D(I))}\prod_{Rez_n\in I_{z_0}^\eta\setminus I}\left( \frac{\rho-\bar{z}_n}{\rho-z_n}\right)\nonumber\\
	&-\frac{1}{4\sqrt{2}\pi }\rho^3A_{12}(z_0,\eta)\overline{M_{12}(\rho ;x,t|D(I))}^3\prod_{Rez_n\in I_{z_0}^\eta\setminus I}\left( \frac{\rho-z_n}{\rho-\bar{z}_n}\right)^3\nonumber\\
	&+\frac{iA_{21}(z_0,\eta)}{2\sqrt{2}(\rho-z_0) }\prod_{Rez_n\in I_{z_0}^\eta\setminus I}\left( \frac{\rho-\bar{z}_n}{\rho-z_n}\right)\left[ M_{12}(\rho ;x,t|D(I))+\overline{M_{11}(\rho ;x,t|D(I))}\right] ;\label{H1}	
	\end{align}
	\begin{align}
	&H_2(\rho,z_0;x,t|D(I))=\nonumber\\
	&-\frac{iA_{12}(z_0,\eta)}{2\sqrt{2}(\rho-z_0) }\prod_{Rez_n\in I_{z_0}^\eta\setminus I}\left( \frac{\rho-\bar{z}_n}{\rho-z_n}\right)\left[ M_{12}(\rho ;x,t|D(I))+\overline{M_{11}(\rho ;x,t|D(I))}\right]\nonumber\\
	&-[M_+(\rho;x,t|D(I))]_{12}^2\prod_{Rez_n\in I_{z_0}^\eta\setminus I}\left( \frac{\rho-\bar{z}_n}{\rho-z_n}\right)^3\frac{1}{4\sqrt{2}\pi } A_{21}(z_0,\eta)\overline{[M_+(\rho;x,t|D(I))]_{11}}\nonumber\\
	&+|[M_+(\rho;x,t|D(I))]_{12}|^2\frac{1}{4\sqrt{2}\pi }\rho A_{12}(z_0,\eta)[M_+(\rho;x,t|D(I))]_{11}\prod_{Rez_n\in I_{z_0}^\eta\setminus I}\left( \frac{\rho-z_n}{\rho-\bar{z}_n}\right)\nonumber\\
	&+\frac{1}{4\sqrt{2}\pi } A_{21}(z_0,\eta)\overline{[M_+(\rho;x,t|D(I))]_{11}}^3\prod_{Rez_n\in I_{z_0}^\eta\setminus I}\left( \frac{\rho-\bar{z}_n}{\rho-z_n}\right)^3\nonumber\\
	&+\overline{[M_+(\rho;x,t|D(I))]_{11}}\frac{1}{4\sqrt{2}\pi }\rho^2A_{12}(z_0,\eta)\overline{[M_+(\rho;x,t|D(I))]_{12}}^2\prod_{Rez_n\in I_{z_0}^\eta\setminus I}\left( \frac{\rho-z_n}{\rho-\bar{z}_n}\right)
	\label{H2}.
	\end{align}
	And from the symmetry of $M^{RHP}$: $M^{RHP}_{22}(z)=\overline{M^{RHP}_{11}(\bar{z})}$, $M^{RHP}_{21}(z)=-z\overline{M^{RHP}_{12}(\bar{z})}$ we obtain the whole result of $M^{RHP}_+(\rho)$
\end{proposition}
\begin{proof}
	\textbf{1.} when $\rho\in U_{z_0}$, from (\ref{transm4}), we have
	\begin{align}
	M^{RHP}_+(\rho)=E(\rho)M_+^{(out)}(\rho),
	\end{align}
	from which we have
	\begin{align}
	[M^{RHP}_+(\rho)]_{11}&=E_{11}(\rho)[M_+^{(out)}(\rho)]_{11}+E_{12}(\rho)[M_+^{(out)}(\rho)]_{21},\\
	[M^{RHP}_+(\rho)]_{12}&=E_{11}(\rho)[M_+^{(out)}(\rho)]_{12}+E_{12}(\rho)[M_+^{(out)}(\rho)]_{22}.
	\end{align}
Combining with (\ref{mout}) and \textbf{proposition \ref{Erho}} we come to the result.
	
	\textbf{2.} when $\rho\notin U_{z_0}$, from (\ref{transm4}) we have
	\begin{equation}
	M^{RHP}_+(\rho)=E(\rho)M_+^{(out)}(\rho)M_+^{pc}(\rho;z_0,\eta),
	\end{equation}
	from which we have
	\begin{align}
	[M^{RHP}_+(\rho)]_{11}=&E_{11}(\rho)[M_+^{(out)}(\rho)]_{11}[M_+^{PC}(\rho)]_{11}+E_{12}(\rho)[M_+^{(out)}(\rho)]_{21}[M_+^{PC}(\rho)]_{11}\nonumber\\
	+&E_{11}(\rho)[M_+^{(out)}(\rho)]_{12}[M_+^{PC}(\rho)]_{21}+E_{12}(\rho)[M_+^{(out)}(\rho)]_{22}[M_+^{PC}(\rho)]_{21},\\
	[M^{RHP}_+(\rho)]_{12}=&E_{11}(\rho)[M_+^{(out)}(\rho)]_{11}[M_+^{PC}(\rho)]_{12}+E_{12}(\rho)[M_+^{(out)}(\rho)]_{21}[M_+^{PC}(\rho)]_{12}\nonumber\\
	+&E_{11}(\rho)[M_+^{(out)}(\rho)]_{12}[M_+^{PC}(\rho)]_{22}+E_{12}(\rho)[M_+^{(out)}(\rho)]_{22}[M_+^{PC}(\rho)]_{22}.
	\end{align}
	Combining with (\ref{mout}), (\ref{asympc}) and \textbf{proposition \ref{Erho}} we come to the result.
\end{proof}

\section{$\bar{\partial}$ Problem}
\quad $\bar{\partial}$-Problem 5 of $M^{(3)}$ is equivalent to the integral equation
\begin{equation}
M^{(3)}(z)=(1,0)+\frac{1}{\pi}\int_C\dfrac{\bar{\partial}M^{(3)}(s)}{z-s}dm(s)=(1,0)+\frac{1}{\pi}\int_C\dfrac{M^{(3)}(s)W^{(3)}(s)}{z-s}dm(s),
\end{equation}
where $W^{(3)}(s)=M^{RHP}(s)\bar{\partial}R^{(2)}(s)M^{RHP}(s)^{-1}$, and $m(s)$ is the Lebegue measure on the $C$. If we denote $C_z$ is the left Cauchy-Green integral  operator,
\begin{equation*}
fC_z(z)=\frac{1}{\pi}\int_C\dfrac{f(s)W^{(3)}(s)}{z-s}dm(s),
\end{equation*}
then
\begin{equation}
M^{(3)}(z)=(1,0)\left(I-C_z \right) ^{-1}.\label{deM3}
\end{equation}
To proof the existence of operator $\left(I-C_z \right) ^{-1}$, we have following Lemma.
\begin{lemma}\label{Cz}
	There exists a constant C such that the operator $C_z$ satisfies that
	\begin{equation}
	\parallel C_z\parallel_{L^\infty\to L^\infty}\leq C|t|^{-1/4}.
	\end{equation}
\end{lemma}
\begin{proof}
	For any $f\in L^\infty$,
	\begin{equation}
	\parallel fC_z \parallel_{L^\infty}\leq \parallel f \parallel_{L^\infty}\frac{1}{\pi}\int_C\dfrac{|W^{(3)}(s)|}{|z-s|}dm(s),
	\end{equation}
	where
	\begin{equation*}
	|W^{(3)}(s)|\leq 	\parallel M^{RHP} \parallel_{L^\infty}|\bar{\partial}R^{(2)} (s)|\parallel M^{RHP} \parallel_{L^\infty}^{-1}\lesssim |\bar{\partial}R^{(2)} (s)|.
	\end{equation*}
	So we only need to  estimate
	\begin{equation*}
	\frac{1}{\pi}\int_C\dfrac{|\bar{\partial}R^{(2)} (s)|}{|z-s|}dm(s).
	\end{equation*}
	We only proof the case $\eta=+1$, the case $\eta=-1$ can be proof in the same way. For $\bar{\partial}R^{(2)} (s)$ is a piece-wise function, we detail the case in the region $\Omega_1$, the other regions are similar. From (\ref{dbarRj}) and (\ref{DBARR2}) we have
	\begin{equation}
	\parallel C_z\parallel_{L^\infty\to L^\infty}\leq C(I_1+I_2+I_3),
	\end{equation}
	where for $s=u+vi$,
	\begin{align}
	&I_1=\int\int_{\Omega_1}\dfrac{|\bar{\partial}X_Z(s)e^{-8tv(u-z_0)}|}{|z-s|}dudv,\hspace{0.5cm}I_2=\int\int_{\Omega_1}\dfrac{|sp_1'(u)e^{-8tv(u-z_0)}|}{|z-s|}dudv,\nonumber\\
	&I_3=\int\int_{\Omega_1}\dfrac{|s-z_0|^{-1/2}e^{-8tv(u-z_0)}}{|z-s|}dudv.
	\end{align}
	First we bound $I_1$, For $z=\alpha+\beta i$, note that
	\begin{equation}
	\parallel (s-z)^{-1}\parallel_{L^2(v+z_0, +\infty)}^2=\int^\infty_{v+z_0}\frac{1}{v-\beta}\left[  \left( \frac{u-\alpha}{v-\beta}\right) ^2+1\right]  ^{-1} d\left( \frac{u-\alpha}{v-\beta}\right) \leq\frac{\pi}{v-\beta},
	\end{equation}
	then we have
	\begin{align}
	I_1&=\int_0^\infty\int^\infty_{v+z_0}\dfrac{|\bar{\partial}X_Z(s)e^{-8tv(u-z_0)}|}{|z-s|}dudv\nonumber\\
	&\leq\int_0^\infty \parallel\bar{\partial}X_Z(s)\parallel_{L^2_u(v+z_0,\infty)}\parallel (s-z)^{-1}\parallel_{L^2(v+z_0, +\infty)} e^{-8tv^2}dv\nonumber\\
	&\lesssim \int_0^\infty\frac{ e^{-8tv^2}}{|v-\beta|^{1/2}}dv\leq t^{-1/4} \int_R\frac{ e^{-8(\sqrt{t}\beta+w)^2}}{|w|^{1/2}}dw\leq C_1 t^{-1/4} .
	\end{align}
	And for $I_2$,
	\begin{align}
	I_2&\leq\int_0^\infty \parallel (u^2+v^2)^{1/2}p_1'(u)\parallel_{L^2_u(v+z_0,\infty)}\parallel (s-z)^{-1}\parallel_{L^2(v+z_0, +\infty)} e^{-8tv^2}dudv\nonumber\\
	&\lesssim \int_0^\infty \parallel p_1(u)\parallel_{H^2_u(R)}\frac{ (1+v)e^{-8tv^2}}{|v-\beta|^{1/2}}dv.
	\end{align}
	And using $e^{-m}\leq m^{-1/4}$ for $m\geq 0$ we obtain
	\begin{align}
	 \int_0^\beta \beta^{-1/2}\frac{ ve^{-8t\beta^{2}(v/\beta)^2}}{\beta|v/\beta-1|^{1/2}}d(v/\beta)&=\int_0^1 \beta^{-1/2}\frac{ we^{-8t\beta^{2}w^2}}{|w-1|^{1/2}}dw\nonumber\\
	 &\lesssim t^{-1/4} \int_0^1  w(1-w)^{-1/2}dw\leq Ct^{-1/4}.
	\end{align}
	Together with
	\begin{align}
	&\int_\beta^\infty  ve^{-8tv^2}(v-\beta)^{-1/2}dv\nonumber\\
	&=\int_0^\infty  t^{-3/4}e^{-8[\sqrt{t}(w+\beta)]^2}(\sqrt{t}w)^{1/2}+\beta t^{-1/4}e^{-8[\sqrt{t}(w+\beta)]}(\sqrt{t}w)^{-1/2}d\sqrt{t}w\nonumber\\
	&\lesssim t^{-1/4},
	\end{align}
	we have that there exists a constant $C_2>0$ such that $I_2\leq C_2t^{-1/4}$.
	For $I_3$, we choose $p>2$ and $q$ H$\ddot{o}$lder conjugate to $p$, and notice that
	\begin{align*}
	\parallel(s-z)^{-1}\parallel_{L^q_u(v+z_0,\infty)}^q&=\int^\infty_{v+z_0}\left[1+\left( \frac{u-\alpha}{v-\beta}\right) ^2 \right]^{-q/2} |v-\beta|^{-q+1} d\left( \frac{u-\alpha}{v-\beta}\right)\\
	&\leq C_q|v-\beta|^{-q+1},\\
	\parallel(s-z_0)^{-1/2}\parallel_{L^p_u(v+z_0,\infty)}^p&=\int^\infty_{v+z_0}\left[1+\left( \frac{u-z_0}{v}\right) ^2 \right]^{-p/4} |v|^{-p/2+1} d\left( \frac{u-z_0}{v}\right)\\
	&\leq C_p|v|^{-p/2+1},
	\end{align*}
	by H$\ddot{o}$lder inequality we have
	\begin{align}
	I_3&\leq\int^\infty_{0}e^{-8tv^2}\parallel(s-z_0)^{-1/2}\parallel_{L^p_u(v+z_0,\infty)}\parallel(s-z)^{-1}\parallel_{L^q_u(v+z_0,\infty)}dv\nonumber\\
	&\leq	max[ C_p,C_q] \int^\infty_{0}e^{-8tv^2}v^{1/p-1/2}|v-\beta|^{1/q-1}dv.
	\end{align}
	And using the same way as estimating $I_2$, we obtain a constant $C_3>0$ such that $I_3\leq C_3t^{-1/4}$. Finally we come to the  result by combining above equations.
\end{proof}
From this Lemma, we obtain that for sufficiently large t, $	\parallel C_z\parallel_{L^\infty\to L^\infty}<1$, so the operator $\left(I-C_z \right) ^{-1}$ exists, which means $M^{(3)}$ unique exist with   property
\begin{equation}
\parallel M^{(3)}\parallel_\infty\lesssim1.\label{normM3}
\end{equation}
To recover the long-time asymptotic behavior of $q(x,t)$ by reconstruction formula, we need to consider the asymptotic behavior of $M^{(3)}_1$, where $M^{(3)}_1$ is given by the Laurent-expansion of $M^{(3)}$ as $z\to\infty$
 \begin{equation}
 M^{(3)}(z)=(1,0)+\frac{M^{(3)}_1(x,t)}{z}+\frac{1}{z\pi}\int_C\dfrac{sM^{(3)}(s)W^{(3)}(s)}{z-s}dm(s),\label{expM3}
 \end{equation}
and
\begin{equation}
M^{(3)}_1(x,t)=\frac{1}{\pi}\int_CM^{(3)}(s)W^{(3)}(s)dm(s).
\end{equation}
Then we start to estimate $M^{(3)}_1$.
\begin{lemma}
	For $z=iy$, $y\in R$ and $y\to+\infty$, we have
	\begin{equation}
	|M^{(3)}_1(x,t)|\lesssim |t|^{-3/4}.\label{M31}
	\end{equation}
\end{lemma}
\begin{proof}
	From \textbf{Lemma \ref{Cz}} and (\ref{deM3}), we have $\parallel M^{(3)}\parallel_\infty \lesssim1$. And we only estimate the integral on $\Omega_1$ since the other estimates are similar. Like in the above Lemma, by (\ref{dbarRj}) and (\ref{DBARR2}) we obtain
	\begin{equation}
	|\frac{1}{\pi}\int_{\Omega_1}M^{(3)}(s)W^{(3)}(s)dm(s)|\lesssim \frac{1}{\pi}\int_{\Omega_1}|W^{(3)}(s)|dm(s)\lesssim I_4+I_5+I_6,
	\end{equation}
	where for $s-z_0=u+vi$
	\begin{align}
	&I_4=\int\int_{\Omega_1}|\bar{\partial}X_Z(s)|e^{-8tvu}dudv,\hspace{0.5cm}I_5=\int\int_{\Omega_1}|sp_1'(u)|e^{-8tvu}dudv,\nonumber\\
	&I_6=\int\int_{\Omega_1}|s-z_0|^{-1/2}e^{-8tvu}dudv.
	\end{align}
	By Cauchy-Schwarz inequality we have
	\begin{align}
	|I_4|\leq& \int_0^\infty\parallel\bar{\partial}X_Z(s)\parallel_{L^2_u(v+z_0,\infty)}\left(\int_{v}^\infty e^{-16tvu}du\right) ^{1/2}dv\nonumber\\
	\lesssim&\int_0^\infty t^{-1/2}v^{-1/2}e^{-8tv^2}dv=\int_0^\infty t^{-3/4}(\sqrt{t}v)^{-1/2}e^{-8(\sqrt{t}v)^2}d(\sqrt{t}v)\nonumber\\
	\leq&C_4t^{-3/4}.
	\end{align}
	And for $I_5$,
	\begin{align}
	|I_5|\leq& \int_0^\infty\parallel\left( (u+z_0)^2+v^2\right) ^{1/2}p_1'\parallel_{L^2_u(v+z_0,\infty)}\left(\int_{v}^\infty e^{-16tvu}du\right) ^{1/2}dv\nonumber\\
	\lesssim &\int_0^\infty v^{-1/2}t^{-1/2}\parallel p_1(u)\parallel_{H^2_u(R)} (1+v)e^{-8tv^2}dv\nonumber\\
	\lesssim &\int_0^\infty t^{-3/4}(\sqrt{t}v)^{-1/2}e^{-8(\sqrt{t}v)^2}+t^{-5/4}(\sqrt{t}v)^{1/2}e^{-8(\sqrt{t}v)^2} d(\sqrt{t}v)\nonumber\\
	\leq&C_5 t^{-3/4}.
	\end{align}
	Finally we consider $I_6$ in the same way as we doing with $I_3$ by H$\ddot{o}$lder inequality with $2<p<4$
	\begin{align}
	|I_6|&\leq\int^\infty_{0}\left(\int_{v}^\infty e^{-8qtvu}du\right) ^{1/q}v^{1/p-1/2}dv\lesssim\int^\infty_{0}t^{-1/q}v^{2/p-3/2}e^{-8tv^2}dv\nonumber\\
	&=\int^\infty_{0}t^{-3/4}(\sqrt{t}v)^{2/p-3/2}e^{-8(\sqrt{t}v)^2}d(\sqrt{t}v)\nonumber\\
	&\leq C_6 t^{-3/4}.
	\end{align}
	These estimates together show the consequence.
\end{proof}
But our eventually aim is to obtain the long-time asymptotic behavior of $u(x,t)$, so we need following estimation about $M^{(3)}(\rho)$.
\begin{proposition}\label{M3rho}
	The unique solution $M^{(3)}$ of $\bar{\partial}$-Problem 5 stratifies
	\begin{equation}
	M^{(3)}(\rho)=(1,0)+O(t^{-3/4}),
	\end{equation}
	for sufficiently large times $|t|>0$, where the implied constant is independent of $t$.
\end{proposition}
\begin{proof}
	\begin{align}
	M^{(3)}(\rho)&=(1,0)+\frac{1}{\pi}\int_C\dfrac{M^{(3)}(s)W^{(3)}(s)}{\rho-s}dm(s),
	\end{align}
	where
	\begin{align}
	&M^{(3)}(s)W^{(3)}(s)=\nonumber\\
	&\left(M^{(3)}_1(s)W^{(3)}_{11}(s)+M^{(3)}_2(s)W^{(3)}_{21}(s), M^{(3)}_1(s)W^{(3)}_{12}(s)+M^{(3)}_2(s)W^{(3)}_{22}(s)\right) .
	\end{align}
	Note that $\bar{\partial} R^{(2)}$ has zeros on its diagonal,  so together with $W^{(3)}=M^{RHP}\bar{\partial} R^{(2)}[M^{RHP}]^{-1}$ we obtain
	\begin{align}
	&W^{(3)}_{11}=\bar{\partial} R^{(2)}_{21}M^{RHP}_{12}M^{RHP}_{22}-\bar{\partial} R^{(2)}_{12}M^{RHP}_{11}M^{RHP}_{21},\\
	&W^{(3)}_{12}=-\bar{\partial} R^{(2)}_{21}[M^{RHP}_{12}]^2+\bar{\partial} R^{(2)}_{12}[M^{RHP}_{11}]^2,\\
	&W^{(3)}_{21}=\bar{\partial} R^{(2)}_{21}[M^{RHP}_{22}]^2-\bar{\partial} R^{(2)}_{12}[M^{RHP}_{21}]^2,\\	
	&W^{(3)}_{22}=-\bar{\partial} R^{(2)}_{21}M^{RHP}_{12}M^{RHP}_{22}+\bar{\partial} R^{(2)}_{12}M^{RHP}_{11}M^{RHP}_{21}.
	\end{align}
	Then using (\ref{normM3}) and (\ref{normMRHp}) to control the size of each term in the integral, and  the symmetry of $M^{RHP}$ we have
	\begin{align}
	 &|M^{(3)}_1(\rho)-1|\lesssim\int_C |\dfrac{s}{\rho-s}\bar{\partial} R^{(2)}_{21}|+|\dfrac{\bar{\partial} R^{(2)}_{12}}{\rho-s}|dm(s)=O(|t|^{-3/4}),\\
	 &|M^{(3)}_0(\rho)|\lesssim\int_C |\dfrac{s^2}{\rho-s}\bar{\partial} R^{(2)}_{21}|+|\dfrac{s}{\rho-s}\bar{\partial} R^{(2)}_{21}|+|\dfrac{\bar{\partial} R^{(2)}_{12}}{\rho-s}|dm(s)=O(|t|^{-3/4})
	\end{align}
	where the last equality of each  estimation we use the same way which used to bound $\int_{C}|W^{(3)}(z)|dm(z)$ in above Lemma to establish the result.
\end{proof}

\section{Long-time asymptotics for modified NLS equation }
\quad Now we begin to consider the long time asymptotics of $q(x,t)$ at first, which is the solution of (\ref{MNLS1}). Inverting the sequence of transformations (\ref{transm1}), (\ref{transm2}), (\ref{transm3}) and (\ref{transm4}), we have
\begin{align}
M(z)=&M^{(3)}(z)M^{RHP}(z)R^{(2)}(z)^{-1}T(z)^{\sigma_3}\nonumber\\
=&M^{(3)}(z)E(z)M^{(out)}(z)R^{(2)}(z)^{-1}T(z)^{\sigma_3},\hspace{0.5cm}\text{when }z \in C\setminus U_{z_0}
\end{align}
To  reconstruct the solution $q(x,t)$, we take $z\to\infty$ along the Straight line $z_0+R^+i$. Then we have that eventually $z\in \Omega_2$, which means $ R^{(2)}(z)=I$. From (\ref{expT}), (\ref{expMout}), (\ref{asyzE}) and (\ref{expM3}), we have
\begin{equation}
M=\left( I+\frac{M_1^{(3)}}{z}+...\right) \left( I+\frac{E_1}{z}+...\right) \left( I+\frac{M_1^{(out)}}{z}+...\right) \left( I+\frac{T_1^{\sigma_3}}{z}+...\right) ,
\end{equation}
which means the coefficient of the $z^{-1}$ in the Laurent expansion of $M$ is
\begin{equation}
M_1=M_1^{(3)}+E_1+M_1^{(out)}+T_1^{\sigma_3}.
\end{equation}
So from (\ref{reconsq}), (\ref{asyqout}), (\ref{expE}) and (\ref{M31}) we have
\begin{equation}
q(x,t)=q_{sol}(x,t;D(I))+|t|^{-1/2}f(x,t)+O(|t|^{-3/4}),\label{resultq}
\end{equation}
where $f(x,t)$ is given in (\ref{f(x,t)}),

Now we begin to construct the solution $u(x,t)$ of (\ref{MNLS})  with initial data $u_0$ by the transformation
\begin{equation}
u(x,t)=q(x,t)e^{-i\int_{-\infty}^x |q(y,t)|^2dy}.
\end{equation}
Because we already have the long time asymptotics of $q(x,t)$, we only need to consider $e^{-i\int_{-\infty}^x |q(y,t)|^2dy}$. From (\ref{intq})and (\ref{arho}), we have
\begin{align}
e^{-i\int_{-\infty}^x |q(y,t)|^2dy}&=\left(\frac{a(\rho)}{\varphi_+^1(\rho)} \right) ^{2}=M^+_{1}(\rho)^{-2}\nonumber\\
&=\left[ M^{(3)}(\rho)M^{RHP}_+(\rho)R_+^{(2)}(\rho)^{-1}T(\rho)^{\sigma_3}\right]_{1} ^{-2}\nonumber\\
&=\left[ M^{(3)}(\rho)M_+^{RHP}(\rho)R_+^{(2)}(\rho)^{-1}\right]_{1} ^{-2}T(\rho)^{-2}\label{eintq}
\end{align}
From the definition of $R^{(2)}$ in Figure \ref{figR2}, we have following situation.

\textbf{$\bullet$} when $x>0$, which means when $\eta=+1$, $\rho<z_0$ or when $\eta=-1$, $\rho>z_0$, $R^{(2)}$ is a upper triangular matrix $\left(\begin{array}{cc}
1 & *   \\
0 & 1
\end{array}\right)$. So
\begin{align}
\left[ M^{(3)}(\rho)M_+^{RHP}(\rho)R_+^{(2)}(\rho)^{-1}\right]_{1} =&[M^{(3)}(\rho)]_{1}[M_+^{RHP}(\rho)]_{11}+[M^{(3)}(\rho)]_{2}[M_+^{RHP}(\rho)]_{21}\nonumber\\
=&[M_+^{RHP}(\rho)]_{11}+O(|t|^{-3/4}).
\end{align}
From Proposition \ref{MRHPrho}, when $\rho\notin U_{z_0}$ we have
\begin{align}
\left[ M^{(3)}(\rho)M_+^{RHP}(\rho)R_+^{(2)}(\rho)^{-1}\right]_{1} &=[M_+(\rho;x,t|D(I))]_{11}\prod_{Rez_n\in I_{z_0}^\eta\setminus I}\left( \frac{\rho-z_n}{\rho-\bar{z}_n}\right)\nonumber\\
&+|t|^{-1/2}H_1(\rho,z_0;x,t|D(I)+O(|t|^{-3/4}),
\end{align}
where $H_1$ is given in (\ref{H1}).
And when $\rho\in U_{z_0}$,
\begin{align}
\left[ M^{(3)}(\rho)M_+^{RHP}(\rho)R_+^{(2)}(\rho)^{-1}\right]_{1}&=[M_+(\rho;x,t|D(I))]_{11}\prod_{Rez_n\in I_{z_0}^\eta\setminus I}\left( \frac{\rho-z_n}{\rho-\bar{z}_n}\right)\nonumber\\
&+|t|^{-1/2}G_1(\rho,z_0;x,t|D(I)+O(|t|^{-3/4}),
\end{align}
where $G_1$ is given in (\ref{G1}).

\textbf{$\bullet$  } when $x<0$, which means when $\eta=+1$, $\rho>z_0$ or when $\eta=-1$, $\rho<z_0$, $R_+^{(2)}(\rho)$ is a lower triangular matrix $\left(\begin{array}{cc}
1 & 0   \\
-R_1(\rho) & 1
\end{array}\right)$, where we note that $\theta(\rho)=0$. So
\begin{align}
\left[ M^{(3)}(\rho)M_+^{RHP}(\rho)R_+^{(2)}(\rho)^{-1}\right]_{1}&=[M^{(3)}(\rho)]_{1}[M_+^{RHP}(\rho)]_{11}+[M^{(3)}(\rho)]_{2}[M_+^{RHP}(\rho)]_{21}\nonumber\\
&-R_1(\rho)\left\lbrace [M^{(3)}(\rho)]_{1}[M_+^{RHP}(\rho)]_{12}+[M^{(3)}(\rho)]_{2}[M_+^{RHP}(\rho)]_{22}\right\rbrace \nonumber\\
&=[M^{(3)}(\rho)]_{1}[M_+^{RHP}(\rho)]_{11}-R_1(\rho) [M^{(3)}(\rho)]_{1}[M_+^{RHP}(\rho)]_{12}\nonumber\\
&+O(|t|^{-3/4}).
\end{align}
From Proposition \ref{MRHPrho}, when $\rho\notin U_{z_0}$ we have
\begin{align}
&\left[ M^{(3)}(\rho)M_+^{RHP}(\rho)R_+^{(2)}(\rho)^{-1}\right]_{1} =\nonumber\\
&[M_+(\rho;x,t|D(I))]_{11}\prod_{Rez_n\in I_{z_0}^\eta\setminus I}\left( \frac{\rho-z_n}{\rho-\bar{z}_n}\right)-R_1(\rho)[M_+(\rho;x,t|D(I))]_{12}\prod_{Rez_n\in I_{z_0}^\eta\setminus I}\left( \frac{\rho-\bar{z}_n}{\rho-z_n}\right)
\nonumber\\
&+|t|^{-1/2}\left[ H_1(\rho,z_0;x,t|D(I)-R_1(\rho)H_2(\rho,z_0;x,t|D(I))\right] +O(|t|^{-3/4}),
\end{align}
where $H_1$ and $H_2$ is given in (\ref{H1}) and (\ref{H2})  respectively.
And when $\rho\in U_{z_0}$,
\begin{align}
&\left[ M^{(3)}(\rho)M_+^{RHP}(\rho)R_+^{(2)}(\rho)^{-1}\right]_{1}=\nonumber\\
&[M_+(\rho;x,t|D(I))]_{11}\prod_{Rez_n\in I_{z_0}^\eta\setminus I}\left( \frac{\rho-z_n}{\rho-\bar{z}_n}\right)-R_1(\rho)[M_+(\rho;x,t|D(I))]_{12}\prod_{Rez_n\in I_{z_0}^\eta\setminus I}\left( \frac{\rho-\bar{z}_n}{\rho-z_n}\right)
\nonumber\\
&+|t|^{-1/2}\left[ G_1(\rho,z_0;x,t|D(I)-R_1(\rho)G_2(\rho,z_0;x,t|D(I))\right] +O(|t|^{-3/4}),
\end{align}
where $G_1$ and $G_2$ is given in (\ref{G1}) and (\ref{G2})  respectively.

Combine the above results and \textbf{proposition \ref{Msolrho}} we have following result.
\begin{theorem}
	Let $u(x,t)$ be the solution  of (\ref{MNLS})  with initial data $u_0=u(x,t=0)\in H^{2,2}(R)$ which has  corresponding scatteriing data $\left\lbrace r,\left\lbrace z_k,c_k\right\rbrace _{k=1}^N \right\rbrace $. Fixed $x_1,x_2,v_1,v_2\in R$ with $x_1\leq x_2$ and $v_1\leq v_2$. Let $I=[-v_2/4,-v_1/4]$, and $z_0=-x/4t+\rho$. Denote $u_{sol}(x,t)$ be the $N(I)$ soliton corresponding to reflection-less scatteriing data $\left\lbrace r\equiv0,\left\lbrace z_k,c_k(I)\right\rbrace _{k=1}^N \right\rbrace$ which given in (\ref{dataI}). As $|t|\to\infty$ with $(x,t)\in C(x_1,x_2,v_1,v_2)$, we have
	
	\textbf{$\bullet$  } when $x>0$, from which we have when $\eta=+1$, $\rho<z_0$ or when $\eta=-1$, $\rho>z_0$.
	
	If $\rho\notin U_{z_0}$ we obtain
	\begin{align}
	u(x,t)=u_{sol}(x,t)S(\rho)\left( 1+F_1(x,t)|t|^{-1/2}\right) +O(|t|^{-3/4}),
	\end{align}
	and if $\rho\in U_{z_0}$ we obtain
	\begin{align}
	u(x,t)=u_{sol}(x,t)S(\rho)\left( 1+F_2(x,t)|t|^{-1/2}\right) +O(|t|^{-3/4}),
	\end{align}
	where $S(\rho)=\prod_{Rez_n\in I_{z_0}^\eta\setminus I}\left( \frac{\rho-\bar{z}_n}{\rho-z_n}\right)^2T(\rho)^2$ is a constant depending on $\rho$, $\eta$, $z_0$  and two sets of scatteriing data,
	\begin{align}
	F_1(x,t)=\frac{f(x,t)}{q_{sol}(x,t)}-2exp\left( -\frac{i}{2}\int_{-\infty}^{x}|q_{sol}(y,t)|^2dy\right)\prod_{Rez_n\in I_{z_0}^\eta\setminus I}\left( \frac{\rho-\bar{z}_n}{\rho-z_n}\right)H_1, \\
	F_2(x,t)=\frac{f(x,t)}{q_{sol}(x,t)}-2exp\left( -\frac{i}{2}\int_{-\infty}^{x}|q_{sol}(y,t)|^2dy\right)\prod_{Rez_n\in I_{z_0}^\eta\setminus I}\left( \frac{\rho-\bar{z}_n}{\rho-z_n}\right)G_1.
	\end{align}
	
	\textbf{$\bullet$  } when $x<0$, from which we have when $\eta=+1$, $\rho>z_0$ or when $\eta=-1$, $\rho<z_0$.
	
	If $\rho\notin U_{z_0}$ we obtain
	\begin{align}
	u(x,t)=u_{sol}(x,t)B(x,t)^2\left( 1+F_3(x,t)|t|^{-1/2}\right) +O(|t|^{-3/4}),\label{longtime1}
	\end{align}
	and if $\rho\in U_{z_0}$ we obtain
	\begin{align}
	u(x,t)=u_{sol}(x,t)B(x,t)^2\left( 1+F_4(x,t)|t|^{-1/2}\right) +O(|t|^{-3/4}),\label{longtime2}
	\end{align}
	where
	\begin{align}
	&B(x,t)=\nonumber\\
	&\left(\prod_{Rez_n\in I_{z_0}^\eta\setminus I}\left( \frac{\rho-\bar{z}_n}{\rho-z_n}\right)+\prod_{Rez_n\in I_{z_0}^\eta\setminus I}\left( \frac{\rho-z_n}{\rho-\bar{z}_n}\right)R_1(\rho)\int_{x}^{+\infty}u_{sol}(y,t)dy \right)^{-1},\\
	&F_3(x,t)=\frac{f(x,t)}{q_{sol}(x,t)}-2B(x,t)\left( H_1-R_1(\rho)H_2\right), \\
	&F_4(x,t)=\frac{f(x,t)}{q_{sol}(x,t)}-2B(x,t)\left( G_1-R_1(\rho)G_2\right).
	\end{align}
\end{theorem}

\hspace*{\parindent}
\\

\end{document}